\documentclass[11pt,letterpaper]{article}
\pdfoutput=1
\usepackage[margin=1in]{geometry}
\usepackage[ascii]{inputenc}
\usepackage{mathtools}
\usepackage[dvipsnames,svgnames]{xcolor}
\usepackage{amsmath,amssymb,amsthm,mathtools}
\usepackage{mathrsfs}
\usepackage{enumitem}
\usepackage[bookmarksnumbered,linktocpage,hypertexnames=false,colorlinks=true,linkcolor=NavyBlue,urlcolor=NavyBlue,citecolor=ForestGreen,anchorcolor=green,breaklinks=true,pagebackref=true,pdfusetitle]{hyperref}
\usepackage[nameinlink,capitalize]{cleveref}
\numberwithin{equation}{section}
\newtheorem{theorem}{Theorem}[section]
\newtheorem{lemma}[theorem]{Lemma}

\newtheorem{corollary}[theorem]{Corollary}
\newtheorem{definition}[theorem]{Definition}
\newtheorem{remark}[theorem]{Remark}
\AddToHook{env/theorem/begin}{\crefalias{theorem}{theorem}}
\AddToHook{env/lemma/begin}{\crefalias{theorem}{lemma}}
\AddToHook{env/proposition/begin}{\crefalias{theorem}{proposition}}
\AddToHook{env/corollary/begin}{\crefalias{theorem}{corollary}}
\AddToHook{env/definition/begin}{\crefalias{theorem}{definition}}
\AddToHook{env/remark/begin}{\crefalias{theorem}{remark}}

\usepackage{macros}
\usepackage{bbm}
\usepackage{quantikz}
\newcommand{\calH}{\mathcal H}
\newcommand{\Ell}{\mathrm L}
\newcommand{\U}{\mathrm U}
\newcommand{\Or}{\mathrm O}
\newcommand{\SO}{\mathrm{SO}}
\newcommand{\eps}{\varepsilon}
\DeclareMathOperator{\Pin}{Pin}
\DeclareMathOperator{\End}{End}
\allowdisplaybreaks[4]
\title{The Grothendieck Game and Self-Testing Continuous Groups}
\author{Alexander Kulpe\thanks{Ruhr University Bochum, LMU Munich \& MCQST, \url{alexander.kulpe@rub.de}} \and Giulio Malavolta\thanks{Bocconi University, \url{giulio.malavolta@unibocconi.it}} \and Simon Schmidt\thanks{Ruhr University Bochum, \url{s.schmidt@rub.de}} \and Michael Walter\thanks{LMU Munich \& MCQST, \url{michael.walter@lmu.de}}}
\date{}

\begin{document}
\maketitle

\begin{abstract}
Self-testing certifies quantum states and measurements from observed correlations alone, without requiring a description of the devices that produce them. In this work we establish, to the best of our knowledge, the first robust self-test of the \emph{full continuous} family of measurements generated by Majorana operators. We achieve this by showing that the Grothendieck game, a two-player nonlocal game in which the players receive directions on a sphere and return binary answers, is a self-test of Majorana observables. That is, any near-optimal performance certifies an entire continuously parameterized family of quantum measurements. We further extend this certification to products of Majorana observables, obtaining a robust self-test of the spin representation of the continuous pin group.
\end{abstract}

\tableofcontents

\section{Introduction}\label{sec:intro}
How can we certify a quantum device when its internal workings are unknown?
Bell inequalities~\cite{Bell1964,CHSH1969} turn this question into an experiment that distinguishes quantum correlations from classical explanations.
\emph{Self-testing} identifies the state and quantum measurements on that state from certain ideal quantum correlations, up to local isometries.
Robust self-testing gives quantitative closeness guarantees for near-ideal correlations.
Since the work of Mayers and Yao~\cite{MayersYao1998,MayersYao2004}, this principle has led to robust certification of entanglement and measurements~\cite{McKagueYangScarani2012,ColadangeloGohScarani2017,SupicBowles2020}, which in turn allow protocols to use uncharacterized devices for certified randomness and key distribution~\cite{PironioEtAl2010,VaziraniVidick2014}, classical verification of quantum computations~\cite{ReichardtUngerVazirani2013}, and the control of entangled provers underlying $\mathrm{MIP}^*=\mathrm{RE}$~\cite{JiEtAl2020}.

Self-testing is often studied in the framework of nonlocal games.
A nonlocal game consists of two (or more) cooperative players, called Alice and Bob, that interact with a referee.
The players receive questions and respond with answers.
The referee decides whether the players win or lose based on the rules of the game.
During the game, the players are not allowed to communicate, but quantum players may share an entangled state and perform local quantum measurements to maximize their winning probability.
In this work we are especially interested in XOR games, where the winning condition depends only on the XOR of the players' answers, which are single bits.

Group theory naturally arises in the context of self-testing:
Many games that are self-tests use relations among the observables comprising an optimal quantum strategy to identify a group representation.
For example, the canonical optimal strategy for the CHSH game \cite{CHSH1969} contains the Pauli matrices $X$ and $Z$, which generate the unique irreducible 2-dimensional representation of the real Pauli group. The Pauli matrices are self-adjoint unitaries and anticommute, that is, it holds $X^2=Z^2=I$ and $XZ+ZX=0$.
More generally, for $n=2m$, there are self-adjoint matrices $\gamma_1,\ldots,\gamma_n$ on $\mathbb C^{2^m}$, satisfying
\begin{align}
\gamma_i\gamma_j+\gamma_j\gamma_i=2\delta_{ij}I,\label{eq:majoranarelations}
\end{align}
which are often called \emph{Majorana (fermion) operators}\footnote{A concrete realization is given by $\gamma_{2k+1} = Z^{\ot k} \ot X \ot I^{\ot m-(k+1)}$, $\gamma_{2k+2} = Z^{\ot k} \ot Y \ot I^{\ot m-(k+1)}$, $0\leq k \leq m-1$, where $X, Y$ and $Z$ are the Pauli matrices.}.
These generate the unique irreducible $2^m$-dimensional representation of a finite group~$C_n$ which in mathematics is usually called the \emph{Clifford group} (not to be confused with the Clifford group in quantum information theory), with presentation
\begin{align}
    C_n = \left\langle J, c_1, \dots, c_n \mid J^2 = c_i^2 = 1, [J,c_i] = 1, c_ic_j = Jc_jc_i \; (i \neq j)\right\rangle.\label{eq:Cliffordgroup}
\end{align}
Games that self-test the Majorana operators have been studied in several works.
In~\cite{slofstra2011lower}, Slofstra shows that near-optimal strategies for an XOR game correspond to approximate representations of a certain algebra associated with the game.
It is proven that any collection of self-adjoint matrices that almost anticommute are already close to matrices in the aforementioned representation of the group~$C_n$.
Additionally, as a natural generalization of the CHSH game, the $\mathrm{CHSH}(n)$ game is introduced.
Ostrev~\cite{ostrev2016} then shows that the Majorana operators~$\gamma_i$ together with the maximally entangled state give an optimal strategy for $\mathrm{CHSH}(n)$, and that the operators in any near-optimal strategy are close to those in the optimal strategy.
This gives a robust self-test for the Majorana operators~$\gamma_i$.
Ji~\cite{ji2013binary} gives a binary constraint system game arising from gluing together several magic square games that likewise tests the Majorana operators (but does not establish robustness).

Majorana operators have also been studied in other self-testing frameworks.
In \cite{IrfanEtAl2020} a contextuality-based approach for certifying two-qubit Majorana operators is presented.
Furthermore, a self-test for Majorana operators from Bell inequalities (which generalize nonlocal games) is given in~\cite{MichalskiEtAl2025}. They study row-orthogonal, column-normalized (ROCN) matrices and show that certain ROCN Bell inequalities give the desired self-tests.

Besides self-tests for representations of the Clifford group, another celebrated self-testing example is the Pauli braiding test~\cite{NatarajanVidick2017}.
It can be interpreted as a self-test of a representation of the finite Weyl--Heisenberg group.
Such rigidity arguments also support verification with a single computationally bounded quantum prover under cryptographic assumptions: compiled CHSH and XOR games~\cite{NatarajanZhang2023,CuiEtAl2026}, protocols based on oblivious state preparation and trapdoor claw-free functions~\cite{BartusekKhurana2025,BachoEtAl2025}, and succinct arguments for QMA~\cite{BartusekLiuMalavolta2026}.

The games we have seen so far have finitely many questions and answers and therefore self-test only finitely many observables.
In this work we are interested in the fundamental question:
\begin{center}
    \emph{Is robust self-testing possible for a continuous family of measurements (by a single game)?}
\end{center}
For this, we have to study games with infinitely many questions.
What are the guarantees that one can hope for in such a setting?
In a finite test, every designated measurement occurs with positive probability, whereas under a continuous distribution of questions, any single question has probability zero:
thus an optimal strategy can be changed arbitrarily on a fixed finite set of questions without affecting its value, and hence no fixed finite set of generators can be certified.
Moreover, even certifying finitely many generators would not by itself constrain the devices' measurements at the remaining questions.
The guarantee one should aim for is therefore an \emph{probabilistic} one:
a robust self-test should certify that the continuous family of observables of a near-optimal strategy approximate, on average or even with high probability over the choice of questions, an ideal implementation, under one pair of local isometries.

To our knowledge, self-testing a continuous family of measurements from a near-optimal strategy for a \emph{single} nonlocal game has not previously been studied.
The closest related work is perhaps~\cite{chen2024all}, that certifies arbitrary real projective measurements by using an iterative self-testing procedure. In the iterative protocol, a previously self-tested strategy is extended to include an additional measurement. They show that this can be done in a way that eventually one can include any real projective measurement in such a self-test.

\subsection{Our Results}
Many games in the literature have a finite set of questions and answers, so it is not immediately clear which game is suitable for certifying a continuous family of measurements. Our starting point is the little Grothendieck inequality, originally proved by Grothendieck~\cite{Grothendieck1953}; see also \cite{Pisier2012}. In the language of nonlocal games, the little Grothendieck theorem bounds the quantum-to-classical bias ratio of any XOR game with positive semidefinite cost matrix by \(\frac{\pi}{2}\). Here, the bias measures how much more likely Alice and Bob are to win than to lose: a strategy with winning probability $p$ has bias $2p-1$. The optimal classical (quantum) bias is the supremum of this quantity over classical (quantum) strategies. Following \cite{Briet2011, briet2011generalized}, the construction establishing the sharpness of the little Grothendieck inequality can be interpreted as a sequence of XOR games with infinitely many questions, whose quantum-to-classical bias ratios converge to \(\frac{\pi}{2}\).
Following \cite{Briet2011}, we call the $n$-th game in this sequence the \emph{Grothendieck game} $G_n$. It proceeds as follows.
The referee samples a question pair \((x,y)\) of vectors $x,y\in S^{n-1}$ in the unit sphere with density \( \abs{\ip{x}{y}} / z_n\) with respect to the uniform product measure and sends question \(x\) to Alice and question \(y\) to Bob. Here, $z_n \sim \sqrt{2/(\pi n)}$ is the normalization constant.
Alice and Bob win if their answers \(a,b\in\{\pm1\}\) satisfy
\(ab=\operatorname{sgn}\langle x,y\rangle\), where $\operatorname{sgn}(.)$ is the sign function.
In other words, their answers should agree when the inner product is positive and differ when it is negative.

In fact, the sharpness proof in \cite{Briet2011, briet2011generalized} shows that $G_n$ is universal in the sense that its optimal strategy gives rise to optimal quantum strategies for any finite positive semidefinite XOR game. For us, the universal nature of $G_n$ makes it a natural candidate for self-testing a continuous family of measurements. Until now, it was studied in the context of finding games that need a high amount of entanglement when played optimally. Surprisingly, its self-testing properties have not been studied in the literature. 

For the Grothendieck game, the optimal classical bias is \(\beta_\mathrm{c}(G_n) = z_n\) and its optimal quantum bias is \(\beta_\mathrm{q}(G_n) = \frac{1}{nz_n}\) \cite{Briet2011}. An optimal strategy is given by a continuous family of Majorana operators, measured by Alice and Bob on a maximally entangled state. 
More precisely, the operators have the following form. For even \(n = 2m \geq 2\), let \(\gamma_1, \dots, \gamma_n\) be the Majorana operators introduced in \Cref{eq:majoranarelations}, and define
\begin{align}
    \gamma(x) = \sum_{i = 1}^n x_i \gamma_i, \quad x \in S^{n-1},\label{eq:lincombMajorana}
\end{align}
where \(S^{n-1}\) denotes the unit sphere in \(\R^n\). The Majorana relations imply \(\gamma(x)^* = \gamma(x)\), \(\gamma(x)^2 = I\), and \(\{\gamma(x), \gamma(y)\} = 2 \ip{x}{y}I\). Thus every \(\gamma(x)\) is a binary observable. Any choice of orthonormal basis of $\R^n$ gives us operators that satisfy the relations in \Cref{eq:majoranarelations}. Thus, varying the orthonormal basis yields a continuous family of equivalent representations of \(C_n\). Self-tests for \(n\) Majorana operators satisfying \Cref{eq:majoranarelations} are known, but certifying the entire continuous family simultaneously has not previously been achieved.

In this paper, we show that $G_n$ indeed has strong self-testing properties.
As mentioned before, an optimal quantum strategy for $G_n$ attaining the quantum bias uses a maximally entangled state, with Alice and Bob measuring \(\gamma(x)\) and \(\overline{\gamma(y)}\) on their respective local Hilbert spaces.
We show that every near-optimal quantum strategy is close to this canonical strategy.

\begin{theorem}[Majorana self-testing, \cref{thm:self-testing}]\label{thm:intro-self-testing}
For \(n=2m\), consider a finite-dimensional strategy for the Grothendieck game $G_n$ with state \(\ket\psi\) and observables \(A(x),B(y)\) for $x,y \in S^{n-1}$. If the strategy achieves bias at least \((1-\varepsilon)\beta_\mathrm{q}(G_n)\), then there exist local isometries \(\Phi_A,\Phi_B\) and an auxiliary state
\(\ket{\mathrm{aux}}\) such that 
    \begin{align*}
        \norm[\big]{(\Phi_A \otimes \Phi_B)\ket{\psi} - \ket{\psi^+_{2^m}} \otimes \ket{\mathrm{aux}}}^2 &\leq \bigO(n^2\varepsilon), \\
        \E_{x \sim \sigma_n} \norm[\big]{ ((\Phi_A A(x) - (\gamma(x) \otimes I)\Phi_A) \otimes I)\ket\psi}^2 &\leq \bigO(n\varepsilon), \\
        \E_{y \sim \sigma_n} \norm[\big]{ (I\otimes (\Phi_B B(y) - (\overline{\gamma(y)} \otimes I)\Phi_B))\ket\psi}^2 &\leq \bigO(n\varepsilon)
    \end{align*}
where \(\ket{\psi^+_{2^m}}\) is the $2^m$-dimensional maximally entangled state.
\end{theorem}
Changing the measurements at any finite set of questions in a strategy leaves its bias for the game unchanged, so the expectation over the questions in the theorem is needed. To certify measurements, it is important that we can obtain self-testing guarantees for fixed operators. In \Cref{cor:HighProb}, we deduce from \Cref{thm:intro-self-testing} that in a Haar-random orthonormal basis \(Q = (q_1, \dots, q_n)\) of \(\R^n\)  all \(2n\) observables \(A(q_i), B(q_i)\) are simultaneously certified with high probability.

\Cref{thm:intro-self-testing} gives us the desired self-test for a continuous family of observables. But can even self-test an irreducible representation of a \emph{continuous group}? 
 The set of operators in \Cref{eq:lincombMajorana} can be extended by taking products of the elements $\gamma(x)$. This yields an irreducible representation of the pin group, which we will describe now. Consider self-adjoint unitaries~$c(x)$ for~$x \in S^{n-1}$ subject to the relation $c(x)c(y)+c(y)c(x)=2\langle x, y \rangle $. The pin group $\Pin(n)$ is generated by the elements $c(x)$ and the representation
 \begin{align*}
    \Gamma \colon \Pin(n) \to \U(\C^{2^m}), \quad c(x_1) \cdots c(x_k) \mapsto \gamma(x_1) \cdots \gamma(x_k), \quad k\in\N, \, x_1,\dots,x_k \in S^{n-1},
\end{align*}
is called the \emph{spin representation} of $\Pin(n)$. Any element $g\in \Pin(n)$ can be written as signed product of at most $n$ generators $c(x)$. 
The main challenge in lifting the self-testing result from the generators to the group is that such
a decomposition is not unique. A priori it is not clear how the operators $A(x)$ and $B(y)$ of Alice and Bob can be extended to an approximate representation of the pin group. This is solved by taking rotations $R \in \SO(n)$ and their associated automorphisms on the pin group into account. We obtain unitaries $A_R(g)$ and $B_R(g)$ for $g \in \Pin(n)$, $R\in \SO(n)$ which are, crucially, implementable by Alice and Bob when given unitary access to their observables. We show that the Grothendieck game is a self-test for the spin representation of the pin group.

\begin{theorem}[Spin representation self-test, \Cref{thm:spinselftest}]\label{thm:intro-spinselftest}
In the situation of \cref{thm:intro-self-testing}, we also have (for the same local isometries and state):
    \begin{align*}
        \E_{R \in \SO(n)} \E_{g \in \Pin(n)} \norm[\big]{((\Phi_A A_R(g) - (\Gamma(g) \otimes I)\Phi_A) \otimes I)\ket\psi}^2 &\leq \bigO(n^3 \varepsilon), \\
        \E_{R \in \SO(n)} \E_{g \in \Pin(n)} \norm[\big]{(I \otimes (\Phi_B B_R(g) - (\overline{\Gamma(g)}\otimes I)\Phi_B))\ket{\psi}}^2 &\leq \bigO(n^3 \varepsilon).
    \end{align*}
\end{theorem}

Rounding an approximate representation of the Clifford group \(C_n\) to an exact one is a fundamental step in the proof of the self-testing statements. 
In \cite{vidicknotes} a suitable state-dependent version of the Gowers--Hatami stability theorem \cite{GowersHatami2017} is given for finite groups.
In \Cref{sect:bonus}, we present a simple and \emph{effective} proof that applies to finite and compact ``continuous'' groups alike, allows targeting a specific irreducible representation, and clarifies what is required so that the isometry can be implemented efficiently.
The isometry is efficient whenever the controlled approximate representation and preparation of a uniform superposition over the group are efficient.
The quantum Fourier transform for the group is needed only to extract a desired irreducible representation.

We also establish rigidity results for classical and vector strategies of the Grothendieck game~$G_n$.
An optimal deterministic strategy is given by the measureable functions \(a(x) =  \sgn\ip{c}{x}\) and \(b(y) = \sgn\ip{c}{y}\), where \(c\in S^{n-1}\). We show in \cref{thm:classical-rigidity} that any near-optimal deterministic strategy is close to such a strategy.
Furthermore, we establish rigidity for vector strategies, which assign real unit vectors \(a(x)\) and \(b(y)\) to the questions, with correlation \(\ip{a(x)}{b(y)}\).
By Tsirelson's theorem~\cite{tsirel1987quantum}, the optimal vector and quantum biases coincide.
In \cref{thm:vector-rigidity}, we prove that every near-optimal vector strategy is close to the vector strategy \(a(x) = x, b(y) = y\).

\subsection{Technical Outline}
We sketch the proof of \cref{thm:intro-self-testing}.
An optimal strategy for Alice and Bob in the Grothendieck game~$G_n$ is to share a maximally entangled state and, upon receiving questions $x$ and $y$, measure the Majorana observables $\gamma(x)$ and $\overline{\gamma(y)}$ (see \Cref{eq:lincombMajorana}), respectively. Our goal is to show that any near-optimal strategy is close to this optimal strategy. 

Consider a near-optimal strategy consisting of a state $\ket{\psi}$ and observables $A(x),B(y)$ for $x,y\in S^{n-1}$. From near-optimality we obtain that, on average, it holds
\begin{align*}
    (A(z)\otimes I)\ket{\psi}\approx \left(\frac{A(x)+A(y)}{\sqrt{2}}\otimes I\right)\ket{\psi}, \quad
     (A(x)\otimes I)\ket{\psi}\approx (I\otimes B(x))\ket{\psi},
\end{align*}
for a Haar-random orthogonal pair $(x,y)$ and $z=\frac{x+y}{\sqrt{2}}$. Those approximate equations routinely imply the anticommutativity relations
\begin{align*}
    (\{A(x),A(y)\}\otimes I)\ket{\psi}\approx 0,\quad
    (I\otimes \{B(x),B(y)\})\ket{\psi}\approx 0,
\end{align*}
for orthogonal pairs $(x,y)$, where we write $\{P,Q\}=PQ+QP$ for the anticommutator.

At this stage, it is tempting to apply the Gowers--Hatami theorem (\Cref{thm:GH irrep}) directly to the pin group $\Pin(n)$. However, we only have observables $A(x)$ and $B(x)$ for its generators $c(x)$, and have not extended them to approximate representations of the whole group. We solve this by choosing a Haar-random orthonormal basis of $\R^n$ and using the corresponding observables to construct approximate representations of the finite Clifford group $C_n$, to which the theorem applies. Randomness is important in the argument, because measurements at a prescribed finite set can be changed without affecting the bias of the strategy. We take this approach to obtain good robustness bounds for the Majorana observables. As we will see, these bounds then allow us to extend the self-test to the whole pin group by taking products of the observables.

Let $Q=(q_i)$ be a Haar-random orthonormal basis of $\R^n$. Define $\tilde A_Q(x)=\sum_i\ip{q_i}{x}A(q_i)$ and similarly for Bob. To make our approach work, we first show
\begin{align}
(\widetilde A_Q(x)\otimes I)\ket{\psi}\approx(A(x)\otimes I)\ket{\psi},\quad 
(I\otimes\widetilde B_Q(x))\ket{\psi}\approx (I\otimes B(x))\ket{\psi},\label{eq:sameoperators}
\end{align}
on average over a Haar-random basis $Q$ and a uniformly random question $x$. Thus, certifying the observables $A(q_i)$ and $B(q_i)$ will identify the whole measurement family.

The $2^{n+1}$ elements of the Clifford group $C_n$ (see \Cref{eq:Cliffordgroup}) have unique normal forms, on which we define
\[
f_A^Q(J^{a_0}c_1^{a_1}\cdots c_n^{a_n})=(-1)^{a_0}A(q_1)^{a_1}\cdots A(q_n)^{a_n},
\]
for $a_0,\dots,a_n\in\{0,1\}$, and similarly $f_B^Q$. We show that these are approximate representations of $C_n$. Averaging over group words allows us to control the errors that arise when reordering the generators. In particular, we obtain
\begin{align*}
    f_A^Q(h)A(q_i)\approx f_A^Q(hc_i), \quad f_A^Q(hJ)=-f_A^Q(h),
\end{align*}
on the state $\ket{\psi}$, on average over $h$ and for a random basis $Q$, and similarly for $f_B^Q$. Now, we can apply the Gowers-Hatami theorem to these approximate representations of the Clifford group $C_n$. For $n=2m$, the Majorana representation $c_i\mapsto\gamma_i$ is the unique irreducible representation with $J\mapsto-I$. Therefore, we obtain local isometries $\Phi_A, \Phi_B$ such that 
\begin{align*}
\Phi_A A(q_i)\approx(\gamma(q_i)\otimes I)\Phi_A,\quad
\Phi_B B(q_i)\approx(\overline{\gamma(q_i)}\otimes I)\Phi_B
\end{align*}
on the state $\ket{\psi}$. By linearity, the approximate equations extend to $\tilde A_Q(x)$ and $\tilde B_Q(x)$, and using \Cref{eq:sameoperators} yields that they are also true for the  operators $A(x)$ and $B(x)$ in the near-optimal strategy.

Finally, the state $\ket{\phi}=(\Phi_A\otimes \Phi_B)\ket{\psi}$ satisfies
$(\gamma(q_i)\otimes I)\ket{\phi}\approx
(I\otimes\overline{\gamma(q_i)})\ket{\phi},$
averaged over $i$. By a standard argument using consistency of the ideal Majorana observables, we deduce
\begin{align*}
(\Phi_A\otimes \Phi_B)\ket{\psi}
\approx
\ket{\psi_{2^m}^+}\otimes\ket{\mathrm{aux}} 
\end{align*}
for some auxiliary state $\ket{\mathrm{aux}}$. 

 We want to extend this self-test to the pin group to prove \Cref{thm:intro-spinselftest}. Fix a decomposition for each element of the pin group into a signed product of at most $n$ generators. The challenge for lifting the self-testing result is that for a Haar-random group element, the unit vectors appearing in the chosen decomposition need not be uniformly distributed on the sphere, so the averaged guarantees for the Majorana observables do not apply directly. We solve this by choosing an independent Haar-random rotation $R$, applying the decomposition in the rotated coordinate system, and rotating the resulting unit vectors back. Replacing the corresponding generators by the players' observables gives the unitaries $A_R(g)$ and $B_R(g)$. Averaging over both the group element and the rotation makes each resulting unit vector uniform on the sphere. We can therefore apply the Majorana self-test to each observable appearing in the product decompositions of $A_R(g)$ and $B_R(g)$. We then use the relation between Alice's and Bob's observables on the shared state to control the error in these products, extending the self-test to the spin representation under the same local isometries.

\section{Preliminaries}\label{sec:preliminaries}
We recall XOR games with measurable question spaces and define the Grothendieck game.

\subsection{XOR Games}
In the following, we will consider two-player nonlocal games with measurable question sets and binary answer sets.
The question spaces may be finite or infinite.
Let \(\mathcal{X}, \mathcal{Y}\) be measurable question spaces and let \(\mu\) be a measure on \(\mathcal{X} \times \mathcal{Y}\).
Let \(\pi \colon \mathcal X \times \mathcal Y \to [0, \infty)\) be a measurable probability density with respect to the measure \(\mu\).
Let \(g\colon \mathcal X \times \mathcal Y \to \set{\pm 1}\) be measurable and define \(V \colon \{\pm 1\}^2 \times \mathcal X \times \mathcal Y \to \set{0,1}\) by \(V(a,b \mid x,y) = \frac{1}{2}(1 + abg(x,y))\).
Such a nonlocal game \({G}\) can be described by the tuple \({G} = (\mathcal{X}, \mathcal{Y}, V, \mu, \pi)\) describing a scenario consisting of two non-communicating players, Alice and Bob, interacting with a referee.

In the game, the referee samples a pair of questions \((x,y) \in \mathcal{X} \times \mathcal{Y}\) according to \(\pi\), sending question \(x\) to Alice and question \(y\) to Bob.
Then, Alice (resp.~Bob) returns answer \(a \in \{\pm 1\}\) (resp.~\(b\in \{\pm 1\}\)) to the referee.
They win if \(V(a,b \mid x,y) = 1\) and lose otherwise.

All information about the game \({G}\) is available to the players before the game begins, allowing them to decide on a strategy beforehand.
However, once the game begins they are not allowed to communicate anymore.
The behavior of the players can be modeled by the probabilities \(p_S(a,b \mid x,y)\) of the answers \(a,b\) determined by the strategy \(S\) given questions \(x,y\).
The probability of winning the game \({G}\) under a strategy \(S\) is given by
\[
    \omega({G}, S) = \E_{(x,y) \sim \pi} \sum_{a,b \in \{\pm 1\}} V(a,b \mid x,y) p_S(a,b \mid x,y)
\]
and the bias is given by
\[
    \beta( G, S) \coloneq 2\omega( G, S) - 1 = \E_{(x,y) \sim \pi} g(x,y) \sum_{a,b} abp_S(a,b \mid x,y).
\]

In this paper, we will consider deterministic, classical, vector, and quantum strategies.
\begin{definition}
    A \emph{deterministic (classical) strategy} \(S=(a,b)\) for an XOR game \({G}\) consists of a pair of measurable functions \(a \colon \mathcal{X} \to \set{\pm 1}, b \colon \mathcal{Y} \to \set{\pm 1}\).
\end{definition}
Upon receiving the questions, the players apply their respective functions on their questions and output the outcomes.
The bias of a deterministic strategy \(S\) is given by
\[
    \beta( G, S) = \E_{(x,y) \sim \pi} g(x,y) a(x)b(y).
\]
The deterministic bias of a game \({G}\) is given by
\[
    \beta_\mathrm{det}( G) = \sup_{S \text{ deterministic strategy}} \beta({G}, S).
\]

\begin{definition}
    A \emph{classical strategy} \(S=(a,b)\) for an XOR game \({G}\) consists of
    \begin{itemize}
        \item a probability space \((\Omega, \mathcal{F}, \nu)\) describing shared randomness independent of the questions, along with
        \item a family of deterministic strategies \(\set{(a_w,b_w)}_{w\in \Omega}\) such that the maps \((w,x) \mapsto a_w(x)\) and \((w,y) \mapsto b_w(y)\) are jointly measurable.
    \end{itemize}
\end{definition}
Upon receiving questions \(x\) and \(y\), Alice and Bob use their shared randomness \(w \in \Omega\), apply their respective functions on their questions and output the outcomes.
The bias of a classical strategy \(S\) for a game \({G}\) is given by
\[
    \beta( G, S) = \E_{(x,y) \sim \pi} \E_{w \sim \nu} g(x,y) a_w(x)b_w(y)
\]
and the \emph{classical} bias of the game is given by
\[
    \beta_\mathrm{cl}( G) = \sup_{S \text{ classical strategy}} \beta({G}, S).
\]
Since for every classical strategy there exists a deterministic strategy with the same or higher winning probability it follows that $\beta_\mathrm{cl}( G) = \beta_\mathrm{det}( G)$.

\begin{definition}
    A \emph{quantum strategy} \(S=(A,B,\ket\psi)\) for an XOR game \({G}\) consists of
    \begin{itemize}
        \item a pure quantum state \(\ket{\psi} \in \mathcal{H}_A \otimes \mathcal{H}_B\) in the tensor product of two finite-dimensional Hilbert spaces \(\mathcal{H}_A\) and \(\mathcal{H}_B\),
        \item measurable functions $A \colon \mathcal X \to \U(\mathcal H_A)$ and $B \colon \mathcal Y \to \U(\mathcal H_B)$ satisfying \(A(x)^2 = B(y)^2 = I\) for all $x \in \mathcal X$ and $y \in \mathcal Y$.
    \end{itemize}
\end{definition}
Upon receiving the questions, the players measure their respective observables and output the outcomes.
The bias of a quantum strategy \(S\) for a game \({G}\) is given by
\[
    \beta({G}, S) = \E_{(x,y) \sim \pi} g(x,y) \bra{\psi}A(x) \otimes B(y)\ket{\psi}
\]
and the \emph{quantum} bias of a game \({G}\) is given by
\[
    \beta_\mathrm{q}( G) = \sup_{S \text{ quantum strategy}} \beta({G}, S).
\]

\begin{definition}\label{def:vectorstrat}
    A \emph{vector strategy} \(S=(a,b)\) for an XOR game \( G\) consists of an integer \(N \geq 1\) and measurable maps \(a \colon \mathcal{X} \to S^{N-1}, b \colon \mathcal Y \to S^{N-1}\) where \(S^{N-1}\) is the unit sphere in \(\R^N\).
\end{definition}
The bias of a vector strategy \(S\)  for a game \( G\) is given by
\[
    \beta({G}, S) = \E_{(x,y) \sim \pi} g(x,y) \ip{a(x)}{b(y)}
\]
and the \emph{vector} bias of the game is given by
\[
    \beta_\mathrm{vec}( G) = \sup_{S \text{ vector strategy}} \beta({G}, S).
\]

\begin{definition}[$\eps$-optimality]
For \(\varepsilon \geq 0\), a classical, quantum or vector strategy \(S\) for an XOR game \(G\) is called \emph{\(\varepsilon\)-optimal} if its bias satisfies
\[
    \beta(G, S) \geq (1-\varepsilon)\beta_t(G)
\]
for \(t \in \{\mathrm{cl},\mathrm{q},\mathrm{vec}\}\), respectively.
\end{definition}

Tsirelson's theorem establishes a correspondence between quantum and vector strategies.
Originally formulated for finite question sets, this correspondence extends directly to continuous question spaces, since the constructions apply pointwise and preserve measurability.

\begin{theorem}[{\cite[Thm.~2.1]{tsirel1987quantum}}]\label{thm:tsirelson}
    Let \(G\) be an XOR game with measurable question spaces~\(\mathcal X\) and~\(\mathcal Y\).
    Every quantum strategy $(A,B,\ket\psi)$ admits a vector strategy $(a,b)$ satisfying
    \[
        \bra\psi A(x) \otimes B(y) \ket\psi = \ip{a(x)}{b(y)}
    \]
    for every \((x,y)\in \mathcal X \times \mathcal Y\), and conversely.
    In particular,
    \[
        \beta_{\mathrm q}(G) = \beta_{\mathrm{vec}}(G).
    \]
\end{theorem}

\subsection{Grothendieck Game}
We recall the definition of a nonlocal game, which we refer to as the \emph{Grothendieck game}~\(G_n\), based on the spherical construction studied in \cite{briet2011generalized,Briet2011}.  Let \(n \geq 1\), let \(\sigma_n\) denote the normalized uniform measure on \(S^{n-1}\), and define
\begin{align}
    z_n = \E_{x,y \sim \sigma_n} \abs{\ip{x}{y}} = \E_{x \sim \sigma_n} \abs{\ip{x}{c}} = \E_{x \sim \sigma_n} \abs{x_1} = \frac{\Gamma(n/2)}{\sqrt \pi \Gamma((n+1)/2)}\label{def:zn}
\end{align}
where \(x,y\) are sampled independently from \(\sigma_n\), \(c \in S^{n-1}\) is arbitrary, and where the second equality follows by rotational invariance.

The game is defined as follows:
\begin{itemize}
    \item \textbf{Questions:} The referee samples questions \( (x,y) \in S^{n-1} \times S^{n-1}\) according to the density \( q_n(x,y) = \frac{1}{z_n} \abs{\ip{x}{y}}\) with respect to \(\sigma_n \otimes \sigma_n\).
    He sends $x$ to Alice and $y$ to Bob.
    \item \textbf{Answers:} Alice (resp.~Bob) responds with \(a \in \set{\pm 1}\) (resp.~\(b \in \set{\pm 1}\)).
    \item \textbf{Win Condition:} They win if and only if \( a \cdot b = \sgn(\ip{x}{y})\) where we set \(\sgn (0) = 1\).
\end{itemize}

In the following sections, we study classical, quantum, and vector strategies for this game. In particular, we show in \Cref{sec:self-test} that it has strong self-testing properties.

\section{Classical Rigidity}
In this section, we study classical strategies for the Grothendieck game. We compute the optimal classical bias of \(G_n\) and show that any optimal deterministic strategy has the following form. It consists of
\[
    a(x) = \sgn \ip{c}{x}, \quad b(y) = \sgn\ip{c}{y},
\]
for some \( c \in S^{n-1}\).

We then prove that this characterization is robust. That is,
every \(\varepsilon\)-optimal deterministic strategy is close to one of these strategies in the sense of \cref{thm:classical-rigidity}.

\subsection{Classical Bias}
Since shared randomness does not increase the optimal bias, it suffices to consider deterministic strategies to compute \(\beta_\mathrm{cl}(G_n)\).
Let \( (a,b)\) be a deterministic strategy and define the maps $T_a=\E_{x\sim\sigma_n}a(x)\ket{x}$ and $T_b=\E_{y\sim\sigma_n}b(y)\ket{y}$.
The bias of this strategy is
\begin{equation}\label{eq:classical-bias}
\begin{aligned}
    \beta(G_n, (a,b)) = \E_{x,y\sim\sigma_n}q_n(x,y)\sgn\ip{x}{y}a(x)b(y) =\frac1{z_n}\E_{x,y\sim\sigma_n}\ip{x}{y}a(x)b(y)
    =\frac1{z_n}\ip{T_a}{T_b}.
\end{aligned}
\end{equation}
By rotational invariance,
\begin{equation}\label{eq:classical-moment}
\begin{aligned}
    \norm{T_a}_2
    =\max_{\norm{c}_2=1}\E_{x\sim\sigma_n}a(x)\ip{c}{x}\le\max_{\norm{c}_2=1}\E_{x\sim\sigma_n}\abs{\ip{c}{x}}=z_n.
\end{aligned}
\end{equation}
The same bound holds for \(T_b\). Hence, by Cauchy--Schwarz, the classical bias of the Grothendieck game is bounded by
\[
    \beta_{\mathrm{cl}}(G_n) = \sup_{a,b} \E_{x,y \sim \sigma_n} q_n(x,y) \sgn{\ip{x}{y}}a(x)b(y)
    =\frac1{z_n}\sup_{a,b}\ip{T_a}{T_b}
    \le\frac1{z_n}\sup_{a,b}\norm{T_a}_2\norm{T_b}_2
    \le z_n
\]
where the suprema range over measurable functions \(a,b\colon S^{n-1}\to\set{\pm1}\).
Thus, a strategy attains bias \(z_n\) if and only if \(T_a=T_b\) and \(\norm{T_a}_2=z_n\).

Suppose a strategy attains bias \(z_n\).
Then, \(T_a = T_b = z_n S\) for the unit vector \(S = T_a / \norm{T_a}_2\).
Further, it holds that
\[
    z_n = \norm{T_a}_2 = \E_{x \sim \sigma_n} a(x) \ip{S}{x} \leq \E_{x \sim \sigma_n} \abs{\ip{S}{x}} = z_n.
\]
In particular, it follows that \(a(x) \ip{S}{x} = \abs{\ip{S}{x}}\) almost everywhere.
Thus, \(a(x) = \sgn{\ip{S}{x}}\) almost everywhere and analogously \(b(y) = \sgn{\ip{S}{y}}\) almost everywhere.
Consequently, a measurable strategy \((a,b)\) is optimal if and only if there exists \(S \in S^{n-1}\) such that
\[
    a(x) = a_S(x) \coloneq  \sgn\ip{S}{x}, b(y) = b_S(y) \coloneq \sgn \ip{S}{y}
\]
for \(\sigma_n\)-almost every \(x\) and \(y\), respectively.
Conversely, rotational invariance gives \(T_{a_S} = T_{b_S} = z_n S\) and thus \(\beta_{\mathrm{cl}}(G_n) = z_n\).
The values of \(a\) and \(b\) on the equator \( \{ u \in S^{n-1} \colon \ip{S}{u} = 0\}\) may be chosen arbitrarily in \(\{\pm 1\}\).
Modifying either answer on any \(\sigma_n\)-null set preserves optimality.

\subsection{Robust Classical Rigidity}
We show that any deterministic strategy that approximately attains the optimal classical bias must be close to an optimal deterministic strategy.

\begin{theorem}\label{thm:classical-rigidity}
  Let $\varepsilon\geq 0$.
  Consider an \(\varepsilon\)-optimal deterministic strategy \(S = (a,b)\) for the Grothendieck game~\(G_n\).
    Then there exists \(c \in S^{n-1}\) such that
    \[
        \Pr_{x \sim \sigma_n}(a(x) \neq \sgn \ip{c}{x}) + \Pr_{y \sim \sigma_n}(b(y) \neq \sgn \ip{c}{y}) \leq 2\sqrt{z_nK_n}\sqrt\varepsilon
    \]
    where \(K_1 = K_2 = 1, K_n = (n-1)z_n \leq \sqrt n\) for all \(n \geq 3\).
\end{theorem}

\begin{proof}
We first prove
\begin{equation}\label{eq:classical-small-ball}
    \Pr_{x \sim \sigma_n}(\abs{x_1}\le\delta)\le K_n\delta,
    \qquad 0\le\delta\le1.
\end{equation}
For \(n=1\), this follows from \(\abs{x_1}=1\). For \(n\ge2\), use the parametrization
\(\Phi(u,v)=(u,\sqrt{1-u^2}v)\), with \(u\in(-1,1)\) and \(v\in S^{n-2}\).
For an orthonormal basis \(v_1,\dots,v_{n-2}\) of \(T_vS^{n-2}\),
\begin{align*}
    \partial_u\Phi&=\paren*{1,-\frac{u}{\sqrt{1-u^2}}v},&
    D\Phi(0,v_i)&=\paren*{0,\sqrt{1-u^2}v_i}.
\end{align*}
These vectors are mutually orthogonal, so
\begin{align*}
    (D\Phi)^T(D\Phi)
    &=\mathrm{diag}\paren*{\frac1{1-u^2},1-u^2,\dots,1-u^2},\\
    \sqrt{\det\paren*{(D\Phi)^T(D\Phi)}}
    &=(1-u^2)^{(n-3)/2}.
\end{align*}
Normalizing and using \(\Gamma((n+1)/2)=\frac{n-1}{2}\Gamma((n-1)/2)\) gives the density of \(x_1\):
\[
    f_n(u)=\frac{(n-1)z_n}{2}(1-u^2)^{(n-3)/2}.
\]
Consequently, for \(0\le\delta\le1\),
\begin{align*}
    \Pr_{x \sim \sigma_n}(\abs{x_1}\le\delta)
    =(n-1)z_n\int_0^\delta(1-t^2)^{(n-3)/2}\,\mathrm{d}t\le\begin{cases}
        \delta,&n=2,\\
        (n-1)z_n\delta,&n\ge3.
    \end{cases}
\end{align*}
For \(n = 2\), the expression is \(2 \arcsin(\delta)/\pi\).
Since \(\arcsin\) is convex on \([0,1]\) and vanishes at zero, \(\arcsin(\delta) \leq \delta\arcsin(1) \leq \pi\delta/2\), so the expression is at most \(\delta\).
For \(n\ge3\), the integrand is at most \(1\). This proves \Cref{eq:classical-small-ball}.
Moreover, \(K_n\ge1\) by setting \(\delta=1\) in \Cref{eq:classical-small-ball}, consequently
\[
    z_n=\int_0^1\Pr_{x \sim \sigma_n}(\abs{x_1}>t)\,\mathrm{d}t
    \ge\int_0^{1/K_n}(1-K_nt)\,\mathrm{d}t
    =\frac1{2K_n}.
\]
For \(n\ge3\), Cauchy--Schwarz also gives
\[
    K_n=(n-1)z_n
    \le(n-1)\sqrt{\E_{x\sim\sigma_n}x_1^2}
    =\frac{n-1}{\sqrt n}\le\sqrt n.
\]
Next, we prove the rigidity statement. The case \(\varepsilon=0\) follows from the characterization of optimal strategies. 
For \(\varepsilon \geq 1\), consider an arbitrary unit vector \(c\).
Since
\[
    \sum_{\tilde c \in \{c, -c\}} \paren*{\Pr_{x\sim \sigma_n} (a(x) \neq \sgn \ip{\tilde c}{x}) + \Pr_{y \sim \sigma_n} (b(y) \neq \sgn \ip{\tilde c}{y})} = 2
\]
there is at least one \(\tilde c \in \{c, -c\}\) such that
\[
    \Pr_{x\sim \sigma_n} (a(x) \neq \sgn\ip{\tilde c}{x}) + \Pr_{y\sim \sigma_n} (b(y) \neq \sgn \ip{\tilde c}{y}) \leq 1 \leq 2\sqrt{z_n K_n}\sqrt{\varepsilon}
\]
since \(z_n K_n \geq 1/2\).
Assume henceforth that \(0<\varepsilon<1\).
By \Cref{eq:classical-bias} and \(\varepsilon\)-optimality it follows that
\[
    (1-\varepsilon)z_n^2 \leq \ip{T_a}{T_b}.
\]
Thus,
\[
    \norm{T_a + T_b}_2^2 = \norm{T_a - T_b}_2^2 + 4 \ip{T_a}{T_b} \geq 4\ip{T_a}{T_b} \geq 4(1-\varepsilon)z_n^2,
\]
and we choose \(c = (T_a + T_b)/\norm{T_a + T_b}_2\).
Then,
\[
    \ip{c}{T_a + T_b} = \norm{T_a + T_b}_2 \geq 2\sqrt{(1-\varepsilon)}z_n.
\]
Since \(a(x),b(y)\in\set{\pm1}\), rotational invariance gives
\begin{equation}\label{eq:classical-weighted-errors}
\begin{aligned}
    \E_{x\sim\sigma_n}\mathbbm{1}_{a(x)\ne\sgn\ip{c}{x}}\abs{\ip{c}{x}} + \E_{y \sim \sigma_n}\mathbbm{1}_{b(y)\ne\sgn\ip{c}{y}}\abs{\ip{c}{y}} &=
    \frac 12 \E_{t \sim \sigma_n} \paren*{2\abs{\ip{c}{t}} - a(t)\ip{c}{t} -b(t)\ip{c}{t}} \\ &= z_n - \frac 12 \ip{c}{T_a + T_b} \leq z_n(1 - \sqrt{1-\varepsilon})
\end{aligned}
\end{equation}
We can also find a lower-bound for the expectations as follows. 
Since \( 0 \leq \abs{\ip{c}{x}} \leq 1\) we can write
\begin{align*}
    \E_{x \sim \sigma_n}\mathbbm{1}_{a(x) \neq \sgn \ip{c}{x}}\abs{\ip{c}{x}} &= \int_0^1 \Pr_{x \sim \sigma_n}(a(x) \neq \sgn \ip{c}{x}, \abs{\ip{c}{x}} > t) \,\mathrm{d}t \\
    &\geq \int_0^{\Pr_{x \sim \sigma_n}(a(x) \neq \sgn \ip{c}{x})/K_n} \paren*{\Pr_{x \sim \sigma_n}(a(x) \neq \sgn\ip{c}{x}) - \Pr_{x \sim \sigma_n}(\abs{\ip{c}{x}} \leq t)}\,\mathrm{d}t \\
    &\geq \int_0^{\Pr_{x \sim \sigma_n}(a(x) \neq \sgn \ip{c}{x})/K_n} \paren*{\Pr_{x \sim \sigma_n}(a(x) \neq \sgn\ip{c}{x}) - K_nt} \,\mathrm{d}t \\
    &= \frac{\Pr(a(x) \neq \sgn\ip{c}{x})^2}{2K_n}
\end{align*}
where in the first inequality we used that \( \Pr_{x \sim \sigma_n}(a(x) \neq \sgn \ip{c}{x}) \leq 1\), \(K_n \geq 1\), and \( \Pr(E \cap \{\abs{\ip{c}{x}} > t\}) = \Pr(E) - \Pr(E \cap \{\abs{\ip{c}{x}}\leq t\})\geq \Pr(E) - \Pr(\{\abs{\ip{c}{x}} \leq t\})\) for any event \(E\) and in the last inequality we use \Cref{eq:classical-small-ball}.
Similarly, we get
\[
    \E_{y \sim \sigma_n} \mathbbm{1}_{b(y) \neq \sgn\ip{c}{y}} \abs{\ip{c}{y}}\geq \frac{\Pr(b(y) \neq \sgn \ip{c}{y})^2}{2K_n}
\]
Using \( (a+b) \leq \sqrt{2a^2 + 2b^2}\) and combining this inequality with \Cref{eq:classical-weighted-errors} gives us
\[
    \Pr_{x \sim \sigma_n}(a(x) \neq \sgn\ip{c}{x}) + \Pr_{y \sim \sigma_n}(b(y) \neq \sgn\ip{c}{y}) \leq 2\sqrt{z_nK_n(1-\sqrt{1-\varepsilon})} \leq 2\sqrt{z_nK_n\varepsilon}.
\]
\end{proof}

For classical strategies, the vector \(c \in S^{n-1}\) may depend on the shared randomness.
For each deterministic component \((a_w, b_w)\), define
\[
    \varepsilon_w \coloneq 1 - \frac{\beta(G_n,(a_w, b_w))}{z_n} \geq 0.
\]
Since the classical strategy is \(\varepsilon\)-optimal, linearity of the bias gives \(\E_{w \sim \nu} \varepsilon_w \leq \varepsilon\).
Applying the theorem to each component and using Jensen's inequality therefore gives
\[
    \E_{w \sim \nu} \paren*{ \Pr_{x \sim \sigma_n} (a_w(x) \neq \sgn \ip{c_w}{x}) + \Pr_{y \sim \sigma_n} (b_w(y) \neq \sgn \ip{c_w}{y})} \leq 2 \sqrt{z_n K_n}\E_{w \sim \nu}\sqrt{\varepsilon_w} \leq 2 \sqrt{z_nK_n}\sqrt{\varepsilon}.
\]

\section{Vector Rigidity}
The vector strategy \(a(x) = x\), \(b(y) = y\) attains bias \(1/(nz_n)\) in the Grothendieck game \(G_n\).
We prove that this is the optimal vector bias and that every optimal vector strategy has the form \(a(x) = Vx\), \(b(y) = Vy\) almost everywhere, for a common isometry \(V\).

We then prove that this characterization is robust:
every \(\varepsilon\)-optimal vector strategy is close  to one of these strategies in the sense of \cref{thm:vector-rigidity}.
We use the polar decomposition to construct a common isometry \(V\) for Alice and Bob.
A quantitative concavity estimate then bounds the error in approximating \(a(x)\) and \(b(y)\) by \(Vx\) and \(Vy\).

\subsection{Tsirelson Bound}

The optimal vector bias of the Grothendieck game $G_n$ is given by
\begin{align*}
    \beta_\mathrm{vec}(G_n) = \frac 1{z_n} \sup_N \sup_{\substack{a,b \colon S^{n-1} \to S^{N-1} \\ a,b \text{ measurable}}} \E_{x,y\sim \sigma_n} \ip{x}{y} \ip{a(x)}{b(y)} = \frac 1{z_n} \sup_N \sup_{\substack{a,b \colon S^{n-1} \to S^{N-1} \\ a,b \text{ measurable}}} \ip{T_a}{T_b}_\mathrm{HS}
\end{align*}
where \(T_a =\E_{x\sim\sigma_n} \ket{a(x)}\bra{x}, T_b = \E_{y \sim \sigma_n} \ket{b(y)}\bra{y}\).
We give a direct proof of the following value, which also follows from \cite[Lemma~4]{briet2011generalized}.

\begin{theorem}\label{thm:vectorupperbound}
For the vector bias $\beta_\mathrm{vec}(G_n)$ it holds that
\begin{align*}
    \beta_\mathrm{vec}(G_n) = \frac 1{nz_n},
\end{align*}
where $z_n$ is given in \Cref{def:zn}.
\end{theorem}

\begin{proof}
Using Cauchy-Schwarz we get \allowdisplaybreaks
\begin{align*}
    \beta_\mathrm{vec}(G_n) \leq \frac 1{z_n} \sup_N \paren*{ \sup_{\substack{a \colon S^{n-1} \to S^{N-1} \\ a \text{ measurable}}} \norm{T_a}_\mathrm{HS}}^2
    &= \frac 1{z_n} \sup_N \paren*{ \max_{\norm{C}_\mathrm{HS} = 1} \sup_{\substack{a \colon S^{n-1} \to S^{N-1} \\ a \text{ measurable}}} \ip{C}{T_a}_\mathrm{HS}}^2 \\
    &= \frac 1{z_n}\sup_N \paren*{ \max_{\norm{C}_\mathrm{HS} = 1} \sup_{\substack{a \colon S^{n-1} \to S^{N-1} \\ a \text{ measurable}}}\E_{x\sim\sigma_n}  \bra{a(x)}C\ket{x}}^2 \\
    &= \frac 1{z_n}\sup_N \paren*{ \max_{\norm{C}_\mathrm{HS}=1}  \E_{x\sim\sigma_n} \norm{C\ket{x}}}^2 \\
     &\leq \frac 1{z_n}\sup_N \paren*{ \max_{\norm{C}_\mathrm{HS}=1}  (\E_{x\sim\sigma_n} \norm{C\ket{x}}^2)^{\frac{1}{2}}}^2\\
     &= \frac 1{z_n}\sup_N \paren*{ \max_{\norm{C}_\mathrm{HS}=1} \frac{1}{\sqrt{n}}\norm{C}_\mathrm{HS}}^2\\
     &=\frac 1{nz_n},
    \end{align*}
where we used \(\E_{x\sim\sigma_n}\norm{Cx}\le(\E_{x\sim\sigma_n}\norm{Cx}^2)^{1/2}=\frac{1}{\sqrt n}\norm{C}_{\mathrm{HS}}\).

To see that this upper bound is attained, one uses the strategy $a,b \colon S^{n-1} \to S^{n-1}$ with \(a(x)=b(x)=x\), which gives \(T_a=T_b=I/n\) and therefore has vector bias \(\frac{1}{nz_n}\).
\end{proof}

\subsection{Quantitative Concavity}
To obtain a quantitative rigidity statement for vector strategies that are close to achieving the optimal value, we will rewrite the upper bound in the proof of \Cref{thm:vectorupperbound}. We first need some definitions.

\begin{definition}[(Relative) Interior]
    Let \(C \subseteq \R^n\) be a set.
    An element \(x \in C\) is called an \emph{interior point} of \(C\) if there exists an \(\varepsilon > 0\) such that \(B_\varepsilon(x) \subseteq C\) where \(B_\varepsilon(x) = \set{y \in \R^n \colon \norm{y-x}_2 < \varepsilon}\)
    denotes the open ball of radius \(\varepsilon\) around \(x\).
    The \emph{relative interior} of \(C\) is the set
    \[
        \operatorname{relint}(C) \coloneq \set{x \in C \colon \text{there exists } \varepsilon >0 \text{ such that }B_\varepsilon(x) \cap \operatorname{aff}(C) \subseteq C},
    \]
    where \( \operatorname{aff}(C)\) denotes the affine hull of \(C\), i.e., the smallest affine set that contains \(C\).
\end{definition}

\begin{definition}[\(\mu\)-strongly concave]
    Let \(C \subseteq \R^n\) be a convex set and let \(\mu > 0\).
    The function \(f \colon C \to \R\) is called \emph{\(\mu\)-strongly concave} if for all \(x,y \in C\) and \(t \in [0,1]\)
    \[
        f((1-t)x + ty) \geq (1-t)f(x)+tf(y) + \frac{\mu}{2}t(1-t)\norm{x-y}_2^2.
    \]
\end{definition}

Using the standard Hessian criterion for strong concavity on \( \operatorname{relint}(C)\) extended to \(C\) by continuity we get the following.

\begin{lemma}\label{lem:hessian-criterion}
    Let \(C \subseteq \R^n\) be a nonempty convex set and let \(\mu > 0\). Let \(f \colon C \to \R\) be continuous on \(C\) and twice-differentiable on \(\operatorname{relint}(C)\) with derivatives taken relative to the affine space \(\operatorname{aff}(C)\).
    If
    \[
        \mathrm{d}^2 f(x)[h,h] \leq -\mu \norm{h}_2^2
    \]
    for all \(x \in \operatorname{relint}(C)\) and \(h \in \operatorname{span}(C-C)\) where \(C - C =\{a-b \mid a,b \in C\}\), then \(f\) is \(\mu\)-strongly concave on \(C\).
\end{lemma}

Consider the simplex \(\Delta_{n-1}=\{p\in\R^n:p_i\ge0,\ \sum_i p_i=1\}\). To obtain a quantitative rigidity statement for vector strategies that are close to achieving the optimal value, we will first rewrite the upper bound in the proof of \Cref{thm:vectorupperbound} as follows  
    \begin{align*}
  \beta_{\mathrm{vec}}(G_n)\leq \frac 1{z_n}\sup_N \paren*{ \max_{\norm{C}_\mathrm{HS}=1}  \E_{x\sim\sigma_n} \norm{C\ket{x}}}^2  = \frac 1{z_n} \paren*{ \max_{p\in \Delta_{n-1}} \E_{x\sim\sigma_n} \sqrt{\sum_{i = 1}^{n}p_i x_i^2}}^2.
\end{align*}
The values \(p_i\) can be interpreted as squares of singular values of the matrix $C$. As $C$ is an $N \times n$ matrix, there are $\mathrm{min}(N,n)$ many singular values and since $\norm{C}_\mathrm{HS}=1$ we have $p\in \Delta_{n-1}$ (for \(N<n\), pad the squared singular values with zeros). 

Consider \(F_n(p) \coloneq \E_{x\sim\sigma_n} \sqrt{\sum_{i = 1}^{n} p_i x_i^2}\).
We will show that $F_n(p)$ has the unique maximizer \(p_\star=(1/n, \dots, 1/n)\), with value \(1/\sqrt n\), and that any $p\in \Delta_{n-1}$ for which $F_n$ is close to this maximum must be close to \(p_\star\).

We first need a routine lemma, which is a fact about the moments of the unit sphere.

\begin{lemma}\label{lem:sphere-moments}
    It holds that
    \begin{align*}
          \E_{x \sim \sigma_n} [x_i^4] &= \frac{3}{n(n+2)}, \quad \E_{x \sim \sigma_n}[x_i^2 x_j^2] = \frac{1}{n(n+2)} &&\text{ for \(i \neq j\)}\\
        \E_{x \sim \sigma_n} x_i^6 &= \frac{15}{n(n+2)(n+4)}, \quad \E_{x \sim \sigma_n} x_i^4 x_j^2 = \frac{3}{n(n+2)(n+4)} &&\text{ for }i\neq j,\\
        \E_{x \sim \sigma_n} x_i^2x_j^2x_k^2 &= \frac{1}{n(n+2)(n+4)} &&\text{ for } i \neq j \neq k \neq i.
    \end{align*}    
\end{lemma}

We next establish the quantitative concavity estimate needed for robust vector rigidity.

\begin{lemma}\label{lem:F_n-upper-bound-new}
    The function \( F_n \colon \Delta_{n-1} \to \R, F_n(p) = \E_{x\sim\sigma_n} \paren*{ \sum_{i = 1}^n p_i x_i^2}^{1/2}\) has a unique maximizer \(p_\star = (1/n, \dots, 1/n)\).
    Moreover, for all \(p \in \Delta_{n-1}\) it holds that
    \[
        F_n(p) \leq \frac 1{\sqrt n} - \frac{(n+4)^{3/2}}{4 \cdot 5^{3/2}n(n+2)}\norm{p - p_\star}_2^2
    \]
    where equality holds for \(p = p_\star\), that is, \(F_n(p_\star) = \frac{1}{\sqrt n}\).
\end{lemma}
\begin{proof}

    First, we show that \(F_n\) is \(\mu_n\)-strongly concave for \(\mu_n = (n+4)^{3/2}/(2\cdot5^{3/2}n(n+2))\).
    Note that \(\Delta_{n-1}\) is a convex set, \(F_n\) is continuous on $\Delta_{n-1}$ and twice-differentiable on \(\operatorname{relint}(\Delta_{n-1})\) where
    \[
        \operatorname{relint}(\Delta_{n-1}) = \set*{p \in \Delta_{n-1} \colon p_i > 0 \text{ for all } 1\leq i \leq n}.
    \]
    Recall that, for a set \(C \subseteq \R^n\), its difference set is \( C - C \coloneq \set{x-y \colon x,y \in C}\).
    For the simplex, we have
    \[
        \operatorname{span}(\Delta_{n-1}-\Delta_{n-1}) = \set*{h \in \R^n \colon \sum_{i = 1}^n h_i = 0}.
    \]
    Let \(p \in \operatorname{relint}(\Delta_{n-1})\) and \(h \in \operatorname{span}(\Delta_{n-1} - \Delta_{n-1})\).
    For \(h \neq 0\), let \(\tau_h\) be the probability measure on \(S^{n-1}\) with density \(\frac{\paren*{\sum_i h_i x_i^2}^2}{\E_{x \sim \sigma_n} \paren*{\sum_i h_i x_i^2}^2}\) with respect to \(\sigma_n\).
    Write
    \begin{align}
        \mathrm{d}^2 F_n(p)[h,h] = -\frac 14 \E_{x \sim \sigma_n} \bkt*{\frac{\paren*{\sum_i h_i x_i^2}^2}{\paren*{\sum_i p_i x_i^2}^{3/2}}} = -\frac 14 \E_{x \sim \sigma_n} \paren*{\sum_i h_i x_i^2}^2 \E_{x \sim \tau_h} \paren*{ \sum_i p_i x_i^2}^{-3/2} \nonumber\\
        \leq -\frac 14 \E_{x \sim \sigma_n} \paren*{\sum_i h_i x_i^2}^2 \paren*{ \E_{x \sim \tau_h} \sum_i p_i x_i^2}^{-3/2}.\label{eq:deriv}
    \end{align}
    Now, we want to bound
    \[
        \E_{x \sim \tau_h} \sum_i p_i x_i^2 = \frac{\E_{x \sim \sigma_n} (\sum_i h_i x_i^2)^2 \sum_i p_ix_i^2}{\E_{x \sim \sigma_n} (\sum_i h_i x_i^2)^2}
    \]
    using \Cref{lem:sphere-moments}.
    By \cref{lem:sphere-moments} we get
    \begin{align*}
        \E_{x \sim \sigma_n} \paren*{\sum_{i=1}^n h_ix_i^2}^2 &= \sum_{i=1}^n h_i^2 \E_{x \sim \sigma_n}x_i^4 + 2 \sum_{1 \leq i <j \leq n} h_i h_j \E_{x \sim \sigma_n} x_i^2 x_j^2 \nonumber\\
        &= \frac{3}{n(n+2)} \norm{h}_2^2 + \frac 2{n(n+2)} \sum_{1 \leq i < j \leq n} h_i h_j.
    \end{align*}
    Because \( \sum_i h_i =0\) we have
    \[
        0 = \paren*{ \sum_{i=1}^n h_i}^2 = \sum_{i = 1}^n h_i^2 + 2 \sum_{1 \leq i < j \leq n}h_i h_j
    \]
    and hence 
    \[
        \sum_{1 \leq i < j \leq n}h_ih_j = - \frac 12 \sum_{i = 1}^n h_i^2 = -\frac 12 \norm{h}_2^2.
    \]
    Furthermore, we have
    \begin{align*}
        \E_{x \sim \sigma_n} \paren*{\sum_i h_i x_i^2}^2 \sum_k p_k x_k^2 &= \sum_k p_k \E_{x \sim \sigma_n} \paren*{ h_k^2 x_k^6 + 2h_k \sum_{i \neq k} h_i x_i^2x_k^4 + \sum_{i \neq k} h_i^2 x_i^4 x_k^2 + 2 \sum_{\substack{i < j \\ i,j \neq k}} h_ih_j x_i^2x_j^2x_k^2} \\
        &= \frac{1}{n(n+2)(n+4)}\sum_k p_k\paren*{15h_k^2 + 6h_k\sum_{i \neq k}h_i + 3 \sum_{i \neq k} h_i^2 + 2 \sum_{\substack{i < j \\ i,j \neq k}} h_i h_j} \\
        &= \frac{2 \norm{h}_2^2 + 8 \sum_k p_k h_k^2}{n(n+2)(n+4)} \\
        &\leq \frac{10}{n(n+2)(n+4)}\norm{h}_2^2
    \end{align*}
    where in the second equality we used \Cref{lem:sphere-moments}, in the third we used \(\sum_{i \neq k} h_i = -h_k, \sum_{i \neq k}h_i^2 = \norm{h}_2^2 - h_k^2, 2\sum_{\substack{i < j \\ i,j \neq k}}h_i h_j = 2h_k^2 - \norm{h}_2^2\) and in the last inequality we used that \(\sum_k p_k h_k^2 \leq \norm{h}_2^2\).
    Thus for \(h \neq 0\), we have
    \[
        \E_{x \sim \tau_h} \sum_i p_ix_i^2 \leq \frac{5}{n+4}
    \]
    and inserting this in \Cref{eq:deriv}, we get
    \[
    \mathrm{d}^2F_n(p)[h,h] \leq -\frac 14 \frac{2}{n(n+2)}\norm{h}_2^2 \paren*{\frac{5}{n+4}}^{-3/2} \leq -\frac{(n+4)^{3/2}}{2 \cdot 5^{3/2}n(n+2)} \norm{h}_2^2
    \]
    for all \(p \in \operatorname{relint}(\Delta_{n-1})\) and \(h \in \operatorname{span}(\Delta_{n-1}-\Delta_{n-1})\).
    Thus, by \cref{lem:hessian-criterion}, \(F_n\) is \(\mu_n\)-strongly concave on \(\Delta_{n-1}\) with \(\mu_n = (n+4)^{3/2}/(2\cdot 5^{3/2}n(n+2))\).
    The function \(F_n\) is invariant under permutations of its coordinates.
Applying Jensen's inequality to the concave function \(G(p)=F_n(p)+(\mu_n/2)\norm{p}_2^2\), we obtain
\[
G(p_{\star})=G\left(\frac{1}{n!}\sum_{\pi}\pi p\right)\geq \frac{1}{n!}\sum_{\pi}G(\pi p)=G(p).
\]
From the definition of $G$, we deduce
\[
F_n(p)\le F_n(p_\star)-(\mu_n/2)(\norm{p}^2-\norm{p_\star}_2^2)=F_n(p_\star)-(\mu_n/2)\norm{p-p_\star}_2^2.
\]
Since \(\mu_n > 0\), this also shows that \(p_\star\) is the unique maximizer.
\end{proof}

Note that equality in \Cref{thm:vectorupperbound} holds if \( \ip{T_a}{T_b}_\mathrm{HS} = \norm{T_a}_\mathrm{HS} \norm{T_b}_\mathrm{HS} = \frac 1n\), that is \(T_a = T_b\) and \(\norm{T_a}_{\mathrm{HS}}=1/\sqrt n\).
Further, it holds that
\[
    \frac{1}{\sqrt n} = \norm{T_a}_\mathrm{HS} = \ip{S_a}{T_a}_\mathrm{HS} = \E_{x \sim \sigma_n} \ip{a(x)}{S_ax} \leq \E_{x \sim \sigma_n} \norm{S_ax} \leq \frac 1{\sqrt n},
\]
where \(S_a=T_a/\norm{T_a}_{\mathrm{HS}}\). In particular, it follows that \( \ip{a(x)}{S_ax} = \norm{S_ax}\) almost everywhere.
Furthermore, if equality in the second inequality holds it follows by \cref{lem:F_n-upper-bound-new} that the squared singular values of \(S_a\) are all \(1/n\).
Thus, \(S_a^\dagger S_a = I/n\).
Consequently, there exists an isometry \( V \colon \R^n \to \R^N\) such that \(S_a = V/\sqrt n\).
Together, this shows that \(a(x) = Vx\) almost everywhere and analogously that \(b(y) = Wy\) almost everywhere for some isometry \(W\).
Furthermore, \(T_a = T_b\) implies that \(V = W\). In the next subsection, we will prove a robust version of the statement we just discussed. 

\subsection{Robust Vector Rigidity}\label{sec:robust-vector-rigidity}
In this subsection we want to show that vector strategies with vector bias close to the optimal bias \( \frac{1}{nz_n}\) are close to the optimal vector strategy $a,b \colon S^{n-1} \to S^{n-1}$ with \(a(x)=b(x)=x\) (up to isometry). 
More precisely, for every near-optimal vector strategy \( (a,b), a,b \colon S^{n-1} \to S^{N-1},\) (without loss of generality \(N \geq n\)) there exists a common isometry \( V \colon \R^n \to \R^N\) such that
\[
    \E_{x\sim\sigma_n} \norm*{a(x) - Vx}^2 + \E_{y\sim\sigma_n} \norm*{b(y) - Vy}^2
\]
is small.
In particular, similarly to classical rigidity where we showed that for an \(\varepsilon\)-optimal deterministic strategy \((a,b)\) with \(0 < \varepsilon < 1\) there exists a unit vector \(c = (T_a + T_b)/\norm{T_a + T_b}_2\) such that \(a(x)\) and \(b(x)\) are close to \(\sgn{\ip{c}{x}}\), we will show that we can extract a good common isometry \(V\) from \((T_a + T_b)/\norm{T_a + T_b}_\mathrm{HS}\) for vector strategies \((a,b)\).

First, we will show that for every measurable function \(v \colon S^{n-1} \to \R^N, \norm{v(x)} \leq 1 \text{ for all } x\) and corresponding normalized operator \(S_v \coloneq T_v/\norm{T_v}_\mathrm{HS}\), $T_v \neq 0$, a scaled version of the isometry from the polar decomposition of \(S_v\) is close to \(S_v\) itself.
In a second step, we will apply this to the measurable function \( (a(x) + b(x))/2\) and combine with near-optimality of the vector strategy \((a,b)\) to obtain robust vector rigidity.

\begin{lemma}\label{lem:different-isometries}
   Let \(N\ge n\) and let \(v \colon S^{n-1} \to \R^{N}\) be measurable with \(\norm{v(x)}\leq 1\) for all \(x \in S^{n-1}\).
   Define \(T_v = \E_{x \sim \sigma_n} \ket{v(x)}\bra{x}, S_v = T_v / \norm{T_v}_\mathrm{HS}\) and assume \(\norm{T_v}_\mathrm{HS} > 0\).
    Then, there exists an isometry \(V \colon \R^n \to \R^N\) such that \( S_v = V\abs{S_v}\) and
    \[
        \norm*{S_v - \frac{1}{\sqrt n}V}_\mathrm{HS}^2
        \leq \frac{4 \cdot 5^{3/2}n^{3/2}(n+2)}{(n+4)^{3/2}}(1 - \sqrt n \norm{T_v}_\mathrm{HS}).
    \]
\end{lemma}

\begin{proof}
    Since \(\norm{S_v}_\mathrm{HS} = 1\) we get
    \[
        \norm{T_v}_\mathrm{HS} = \ip{T_v}{S_v}_\mathrm{HS} = \E_{x \sim \sigma_n} \ip{v(x)}{S_vx} \leq \E_{x \sim \sigma_n} \norm{S_v x} \leq \frac{1}{\sqrt n}
    \]
    where we used \(\E_{x \sim \sigma_n} \norm{S_vx} \leq (\E_{x \sim \sigma_n} \norm{S_vx}^2)^{1/2} = \frac{1}{\sqrt n} \norm{S_v}_\mathrm{HS}\).

    Let \(p_v \in \Delta_{n-1}\) denote the squared singular values of \(S_v\).
    By \cref{lem:F_n-upper-bound-new} it follows that
    \[
    \norm{T_v}_\mathrm{HS} \leq \E_{x \sim \sigma_n} \norm{S_v x} = F_n(p_v) \leq \frac 1{\sqrt n} - \frac{(n+4)^{3/2}}{4 \cdot 5^{3/2}n(n+2)}\norm{p_v - p_\star}_2^2
    \]
    where \(p_\star = (1/n, \dots, 1/n)\).
    Hence, it follows that
    \[
        \norm{p_v - p_\star}_2^2 \leq \frac{4 \cdot 5^{3/2} \sqrt n(n+2)}{(n+4)^{3/2}}(1 - \sqrt n \norm{T_v}_\mathrm{HS}).
    \]

    Let \(S_v = U \abs{S_v}\) be the polar decomposition of \(S_v\). 
    Since \(N \geq n\), the restriction of the partial isometry \(U\) to \((\ker S_v)^\perp\) can be extended to an isometry \(V \colon \R^n \to \R^N\) that still satisfies \( S_v = V \abs{S_v}\).
    It follows that
    \begin{align*}
        \norm*{S_v - \frac{1}{\sqrt n}V}_\mathrm{HS}^2
        &= \norm*{ \abs{S_v} - \frac{1}{\sqrt n} I}_\mathrm{HS}^2\\
        &= \sum_i \paren*{ \sqrt{(p_v)_i} - \frac{1}{\sqrt n}}^2 \\
        &\leq n \sum_i \abs*{(p_v)_i - \frac 1n}^2 \\
        &= n \norm{p_v - p_\star}_2^2\\ 
        &\leq \frac{4 \cdot 5^{3/2}n^{3/2}(n+2)}{(n+4)^{3/2}}(1 - \sqrt n \norm{T_v}_\mathrm{HS}).      
    \end{align*}
\end{proof}

We are now ready to state our vector rigidity theorem.

\begin{theorem}\label{thm:vector-rigidity}
    Let \(N \geq n\) and \(\varepsilon \geq 0\).
    Consider an \(\varepsilon\)-optimal vector strategy \(S = (a,b)\) for the Grothendieck game~\(G_n\).
    Then there exists an isometry \(V \colon \R^n \to \R^N\) such that 
    \begin{align*}
        \E_{x\sim \sigma_n} \norm*{a(x) - V x}^2 + \E_{y\sim\sigma_n} \norm*{b(y) - V y}^2
        \leq (20\sqrt5n +2)\varepsilon
        = \bigO(n\varepsilon).
    \end{align*}
\end{theorem}

\begin{proof}
For \(\varepsilon \geq 1\), any isometry gives
\[
    \E_{x\sim \sigma_n} \norm*{a(x) - V x}^2 + \E_{y\sim\sigma_n} \norm*{b(y) - V y}^2 \leq 8 \leq (20\sqrt 5 n+2)\varepsilon
\]
Assume henceforth that \( 0 \leq \varepsilon < 1\).
Consider the vector strategy \( (a,b)\) with bias
\[
    \beta(G_n,(a,b)) = \frac{1}{z_n} \E_{x,y \sim \sigma_n} \ip{x}{y} \ip{a(x)}{b(y)} = \frac{1}{z_n} \ip{T_a}{T_b}_\mathrm{HS}
\]
and consider the measurable map \(v \colon S^{n-1} \to \R^N, v(x) = \frac{a(x) + b(x)}2, \norm{v(x)}\leq 1\) with corresponding operator \(T_v = \frac{T_a + T_b}2\).
By \(\varepsilon\)-optimality it follows that
\begin{equation}
    \norm*{T_v}_\mathrm{HS}^2 = \norm*{\frac{T_a + T_b}2}_\mathrm{HS}^2 = \ip{T_a}{T_b}_\mathrm{HS} + \frac 14 \norm{T_a - T_b}_\mathrm{HS}^2 \geq \ip{T_a}{T_b}_\mathrm{HS} \geq (1-\varepsilon)\frac 1n.\label{eq:T-norm-squared}
\end{equation}
Choose the isometry \(V\) from \cref{lem:different-isometries} and set \(K_n=\frac{4\cdot5^{3/2}n^{3/2}(n+2)}{(n+4)^{3/2}} \leq 20\sqrt5 n\).
The following calculation completes the proof
\begin{align*}
    \E_{x\sim \sigma_n} \norm{a(x) - Vx}^2 + \E_{x \sim \sigma_n} \norm{b(x) - Vx}^2 &= 4 - 2 \ip{T_a + T_b}{V}_\mathrm{HS} \\
     &=4 - 4\norm{T_v}_\mathrm{HS}\ip{S_v}{V}_\mathrm{HS} \\
&=4(1-\sqrt{n}\norm{T_v}_\mathrm{HS})+2\sqrt{n}\norm{T_v}_\mathrm{HS} \norm*{S_v - \frac{1}{\sqrt n} V}_\mathrm{HS}^2\\
&\leq 4(1-\sqrt{n}\norm{T_v}_\mathrm{HS})+2K_n\sqrt{n}\norm{T_v}_\mathrm{HS} (1-\sqrt{n}\norm{T_v}_\mathrm{HS})\\
&=4(1-n\norm{T_v}_\mathrm{HS}^2)+(2K_n-4)\sqrt{n}\norm{T_v}_\mathrm{HS} (1-\sqrt{n}\norm{T_v}_\mathrm{HS})\\
&\leq (K_n+2)(1-n\norm{T_v}_\mathrm{HS}^2)\\
     &\leq (20\sqrt 5 n+2)\varepsilon
\end{align*}
where in the first inequality we used \cref{lem:different-isometries}, in the second inequality we used
\[
    \sqrt n \norm{T_v}_\mathrm{HS}(1-\sqrt{n} \norm{T_v}_\mathrm{HS}) \leq \frac{1 + \sqrt{n} \norm{T_v}_\mathrm{HS}}2(1-\sqrt n \norm{T_v}_\mathrm{HS}) = \frac{1 - n\norm{T_v}_\mathrm{HS}^2}2
\]
since \(0 \leq \sqrt n\norm{T_v}_\mathrm{HS} \leq 1\) and in the last inequality we used \Cref{eq:T-norm-squared}.
\end{proof}

\begin{remark}
Among bounds that hold uniformly over all vector strategies and are linear in \(\varepsilon\), the linear dependence on \(n\) in \cref{thm:vector-rigidity} is optimal.
\end{remark}
\begin{proof}
Let \(n \geq 3\), and let \(Q \colon \R^n \to \R^n\) be the orthogonal projection onto the first \(n-2\) coordinates.
We use the normalized projection strategy from \cite[Lemma~4]{briet2011generalized}:
\begin{equation}\label{eq:projection-strategy}
    a(x) = b(x) = \frac{Qx}{\norm{Qx}}, \quad Qx \neq 0
\end{equation}
and set \(a(x) = b(x) = e_1\) on the null set where \(Qx = 0\).

By \cite[Equation~(27)]{briet2011generalized}, with \(k = n-2\),
\[
    \E_{x \sim \sigma_n} \norm{Qx} = \frac{n-2}{n-1}.
\]
Coordinate reflections and symmetry among the first \(n-2\) coordinates therefore give
\begin{equation}\label{eq:optimality-self-testing}
    \E_{x \sim \sigma_n} \frac{x_j x_k}{\norm{Qx}} = \frac{\delta_{j,k}}{n-2}\E_{x \sim \sigma_n} \norm{Qx} = 
    \frac{\delta_{j,k}}{n-1} ,\quad 1 \leq j \leq n-2, 1 \leq k \leq n
\end{equation}
where the quotient is defined to be zero when \(Qx = 0\).
It follows that \(T_a = T_b = Q/(n-1)\), and hence
\begin{equation}\label{eq:projection-bias}
    \beta(G_n, (a,b)) = \frac{n-2}{z_n(n-1)^2} = \frac{1}{nz_n} \left( 1 - \frac{1}{(n-1)^2}\right).
\end{equation}
Thus, this strategy is \((n-1)^{-2}\)-optimal.

    For every isometry \(V\colon\R^n\to\R^n\), we have
    \begin{align*}
    \E_{x\sim\sigma_n}\norm{a(x)-Vx}^2
    +\E_{y\sim\sigma_n}\norm{b(y)-Vy}^2
    &=4-4\ip{T_a}{V}_{\mathrm{HS}}
    =4-\frac{4}{n-1}
      \sum_{j=1}^{n-2}\ip{e_j}{Ve_j}\\
    &\geq\frac{4}{n-1},
\end{align*}
where equality holds for \(V=I\).
Hence, any uniform upper bound \(f(n)\varepsilon\)
must satisfy \(\frac{4}{n-1}\leq\frac{f(n)}{(n-1)^2}\) proving \(f(n)=\Omega(n)\).
\end{proof}

\section{Self-Testing}\label{sec:self-test}

For even \(n = 2m\), the Majorana observables \(\gamma(x)\) and \(\overline{\gamma(y)}\), measured on the maximally entangled state \(\ket{\psi_{2^m}^+}\), realize the optimal vector strategy \(a(x) = x\), \(b(y) = y\).
By \cref{thm:tsirelson,thm:vectorupperbound}, this strategy attains the optimal quantum bias \(1/(nz_n)\).
We show that every near-optimal quantum strategy is close to this ideal strategy up to local isometries, as stated in \cref{thm:self-testing}. We then extend this result to the certify all operators in the spin representation of the pin group, see \Cref{thm:spinselftest}.

We write \(\E_{s\in S}\) for the uniform average over a nonempty finite set \(S\), and \(\E_{g\in G}\) for expectation with respect to normalized Haar measure on a compact group \(G\).

\subsection{State-Dependent Seminorms}\label{subsec:state-dependent-seminorms}

State-dependent seminorms measure operator errors on the shared state. 
We introduce the notation and record the identities used below.

For a density operator \(\rho\) on a finite-dimensional Hilbert space \(\mathcal H\) and a linear operator \(T \colon \mathcal H \to \mathcal K\), define
\begin{equation}\label{eq:state-dependent-seminorm}
    \norm{T}_\rho^2 \coloneq \tr[T^\dagger T \rho] = \norm{T \rho^{1/2}}_\mathrm{HS}^2.
\end{equation}
For a quantum strategy \(S = (A,B,\ket\psi)\), define the reduced states by
\[
    \rho_A = \tr_B \ket\psi \bra\psi, \quad \rho_B = \tr_A \ket\psi\bra\psi.
\]
Consider linear operators \(X\) and \(Y\) with domains \(\mathcal H_A\) and \(\mathcal H_B\), respectively. Then the seminorms satisfy
\begin{equation}\label{eq:seminorm-purification}
    \norm{X}_{\rho_A} = \norm{(X \otimes I)\ket\psi}, \quad \norm{Y}_{\rho_B} = \norm{(I \otimes Y)\ket\psi}.
\end{equation}
For linear operators \(T \colon \mathcal H \to \mathcal K\) and \(W \colon \mathcal K \to \mathcal L\)
    \[
        \norm{WT}_\rho \leq \norm{W}\norm{T}_\rho,
    \]
where \(\norm{W}\) is the operator norm. If \(W\) is an isometry, then \(\norm{WT}_\rho = \norm{T}_\rho\).

\begin{lemma}\label{lem:state-dependent-seminorm}
    Let \(\rho\) be a density operator on \(\mathcal H\) and consider linear operators \(T_1, \dots, T_n \colon \mathcal H \to \mathcal K\). It holds
    \begin{equation}\label{eq:seminorm-sphere}
        \E_{x \sim \sigma_n} \norm*{ \sum_{i = 1}^n x_i T_i}_\rho^2 = \frac 1n \sum_{i= 1}^n \norm{T_i}_\rho^2
    \end{equation}
\end{lemma}

\begin{proof}
   From \(\E_{x \sim \sigma_n} x_i x_j = \delta_{i,j}/n\), we obtain
    \begin{align*}
        \E_{x \sim \sigma_n} \norm*{ \sum_{i = 1}^n x_i T_i}_\rho^2 &= \sum_{i,j = 1}^n\E_{x \sim \sigma_n} x_ix_j \tr[T_i^\dagger T_j\rho] = \frac 1n \sum_{i = 1}^n \norm{T_i}_\rho^2.
    \end{align*}
\end{proof}

\subsection{Robust Majorana Self-Test}\label{sec:robust-majorana-self-test}

We first show that given a near-optimal quantum strategy for the Grothendieck game~\(G_n\) we can find a linear map~\(L\) such that, on average over \(x\), both \((A(x) \otimes I)\ket\psi\) and \((I \otimes B(x))\ket\psi\) are close to \(Lx\).
The resulting error bound is independent of \(n\).

\begin{lemma}\label{lem:common-linear-map}
    Let \(0 \leq \varepsilon < 1\).
    Consider an \(\varepsilon\)-optimal quantum strategy \( S = (A,B,\ket\psi) \) for the Grothendieck game \(G_n\).
    Then there exists a real-linear map \( L \colon \R^n \to \mathcal{H}_A \otimes \mathcal{H}_B\) such that
    \[
        \E_{x \sim \sigma_n}\paren*{ \norm{(A(x) \otimes I)\ket{\psi} - Lx}^2 + \norm{(I \otimes B(x))\ket{\psi} - Lx}^2} \leq 2\varepsilon.
    \]
    Here, we regard \(\mathcal{H}_A \otimes \mathcal{H}_B\) as a real Hilbert space with inner product \( (u,v)\coloneq \operatorname{Re}\ip{u}{v}\).
\end{lemma}
\begin{proof}
Define \(a(x) = (A(x) \otimes I)\ket{\psi}, b(y) = (I \otimes B(y)) \ket{\psi}\) and \(L = \frac n2 (T_a+ T_b)\) where
\[
    T_a = \E_{x \sim \sigma_n} \ket{a(x)}\bra{x}, \quad T_b = \E_{y \sim \sigma_n} \ket{b(y)}\bra{y}.
\]
Note that by expanding the bias \(\beta(G_n, S) = \operatorname{Re}\ip{T_a}{T_b}_\mathrm{HS}/z_n\) it follows that
\begin{equation} \label{eq:eps-opt-re}
\operatorname{Re}\ip{T_a}{T_b}_\mathrm{HS} \geq (1-\varepsilon)/n.
\end{equation}
    Using \(\E_{x \sim \sigma_n} \ket{x}\bra{x} = I/n\), we obtain
    \begin{align*}
        \E_{x \sim \sigma_n} \norm{(A(x) \otimes I)\ket{\psi} - Lx}^2  &= 1 + \E_{x\sim \sigma_n} \norm{Lx}^2 - 2\E_{x\sim \sigma_n} \operatorname{Re}\ip{a(x)}{Lx} \\
        &= 1 + \frac1n \norm{L}_\mathrm{HS}^2 - 2 \operatorname{Re}\ip{T_a}{L}_\mathrm{HS}.
    \end{align*}
    The analogous identity holds for Bob.
    Adding these identities and substituting \(L = \frac n2(T_a + T_b)\) gives
    \begin{align*}
        \E_{x \sim \sigma_n} \norm{a(x) - Lx}^2 + \E_{y \sim \sigma_n} \norm{b(y) - Ly}^2 &= 2 +\frac n2 \norm{T_a + T_b}_\mathrm{HS}^2 - 2\operatorname{Re}\ip{T_a + T_b}{L}_\mathrm{HS} \\
        &=2 - \frac n2 \norm{T_a + T_b}_\mathrm{HS}^2 \\
        &\leq 2 - 2n\operatorname{Re}\ip{T_a}{T_b}_\mathrm{HS} \\
        &\leq 2\varepsilon
    \end{align*}
    where the first inequality follows from
    \[
        \norm{T_a + T_b}_\mathrm{HS}^2 = 4 \operatorname{Re}\ip{T_a}{T_b}_\mathrm{HS} + \norm{T_a - T_b}_\mathrm{HS}^2 \geq 4\operatorname{Re} \ip{T_a}{T_b}_\mathrm{HS}
    \]
    and the last from \Cref{eq:eps-opt-re}.
\end{proof}

For an orthonormal basis \(Q = (q_1, \dots, q_n)\) of \(\R^n\), define the linear extensions
\[
    \tilde A_Q(x) \coloneq \sum_{i = 1}^n \ip{q_i}{x}A(q_i), \quad  \tilde B_Q(x) \coloneq \sum_{i = 1}^n \ip{q_i}{x}B(q_i).
\]
The next lemma shows that \(\tilde A_Q(x)\) approximates \(A(x)\) in the state-dependent seminorm, on average over \(x\) and a Haar-random basis \(Q = (q_1, \dots, q_n)\), and the same is true for \(\tilde B_Q(x)\) and \(B(x)\). 
Here and in the following we identify elements of the orthogonal group $O(n)$ with orthonormal bases of $\R^n$.

\begin{lemma}\label{lemma:technical-lemma-3}
    Let \(0 \leq \varepsilon < 1\).
    Consider an \(\varepsilon\)-optimal quantum strategy \( S = (A,B, \ket{\psi})\) for the Grothendieck game \(G_n\).
    Then
    \[
        \E_{Q \in O(n)} \E_{x\sim \sigma_n}\paren*{  \norm{A(x) - \tilde A_Q(x)}_{\rho_A}^2 + \norm{B(x) - \tilde B_Q(x)}_{\rho_B}^2 }\leq 4\varepsilon.
    \]
\end{lemma}
\begin{proof}
Define \(a(x) = (A(x) \otimes I)\ket{\psi}\) and \(b(y) = (I \otimes B(y))\ket{\psi}\) with linear extensions \(\tilde  a_Q(x) = \sum_{i = 1}^n \ip{q_i}{x} a(q_i)\) and \(\tilde b_Q(y) = \sum_{i = 1}^n \ip{q_i}{y}b(q_i)\).
As in the proof of \cref{lem:common-linear-map}, we let \(T_a = \E_{x \sim \sigma_n} \ket{a(x)}\bra{x}\), \(T_b = \E_{y \sim \sigma_n} \ket{b(y)}\bra{y}\).
By \Cref{eq:seminorm-purification}, we have
\begin{align*}
    \E_{x \sim {\sigma_n}}\norm*{A(x) - \tilde A_{Q}(x)}_{\rho_A}^2 =  
    \E_{x \sim \sigma_n} \norm{a(x) - \tilde a_Q(x)}^2 &= 1 + \E_{x \sim \sigma_n} \norm{\tilde a_Q(x)}^2 - 2 \operatorname{Re}\E_{x \sim \sigma_n}\ip{a(x)}{\tilde a_Q(x)} \\
    &= 2 - 2\operatorname{Re}\sum_{i = 1}^n \ip{T_a q_i}{a(q_i)}.
\end{align*}
Here the last equality uses \Cref{eq:seminorm-sphere}, where we have \(\E_{x \sim \sigma_n}\norm{\tilde a_Q(x)}^2 = 1\), since \(\norm{A(q_i)}_{\rho_A}^2 = 1\).
The analogous identity holds for Bob.
Adding the two identities and taking the expectation over \(Q\) gives
\begin{align*}
    \E_{Q \in O(n)} \E_{x \sim \sigma_n}\paren*{ \norm{A(x) - \tilde A_Q(x)}_{\rho_A}^2 + \norm{B(x) - \tilde B_Q(x)}_{\rho_B}^2} &= 4 - 2n\paren*{\norm{T_a}_\mathrm{HS}^2 + \norm{T_b}_\mathrm{HS}^2} \\
    &\leq 4 - 4n\operatorname{Re}\ip{T_a}{T_b}_\mathrm{HS} \\
    &\leq 4\varepsilon
\end{align*}
where the last inequality follows from \Cref{eq:eps-opt-re}.
\end{proof}

For \(n = 1\), the Grothendieck game has classical and quantum bias \(1\), attained by \(a(x) = x\) and \(b(y) = y\).
We assume \(n \geq 2\) below, so that Haar-random orthonormal pairs are defined.
We now use \cref{lem:common-linear-map} to derive consistency and anticommutation estimates, averaged over questions and orthonormal pairs, respectively.

\begin{lemma}\label{lemma:average-clifford-relations}
    Let \(n \geq 2\) and let \( 0 \leq \varepsilon < 1\).
    Consider an \(\varepsilon\)-optimal quantum strategy \( S =(A,B, \ket{\psi})\) for the Grothendieck game~\(G_n\).
    Then the following averaged consistency and anticommutation bounds hold:
    \begin{align}
        \E_{x \sim \sigma_n} \norm{(A(x) \otimes I - I \otimes B(x))\ket{\psi}}^2 &\leq \bigO(\varepsilon)\label{eq:consistency} \\
        \E_{x \perp y} \paren*{\norm{\{A(x),A(y)\}}_{\rho_A}^2 + \norm{\{B(x),B(y)\}}_{\rho_B}^2}  &\leq \bigO(\varepsilon) \label{eq:anti-commutation}
    \end{align}
    where \(\E_{x \perp y}\) denotes expectation over a Haar-random ordered orthonormal pair~\((x,y)\).
\end{lemma}

\begin{proof}
    Let \(L\) be the map from \cref{lem:common-linear-map}.
    The triangle inequality and the bound in \cref{lem:common-linear-map} give
    \begin{align*}
        \E_{x \sim \sigma_n} \norm{(A(x) \otimes I - I \otimes B(x))\ket{\psi}}^2 &\leq 2 \E_{x \sim \sigma_n}\paren*{ \norm{(A(x) \otimes I)\ket{\psi} - Lx}^2 + \norm{(I \otimes B(x))\ket{\psi} - Lx}^2} \\
        &\leq 4\varepsilon.
    \end{align*}

    To prove the anticommutation bound, fix an orthonormal pair \((x,y)\) and set \(z \coloneq \frac{x+y}{\sqrt 2}\).
    By \Cref{eq:seminorm-purification}, the triangle inequality, and linearity of \(L\),
    \begin{align*}
        \norm*{A(z) - \frac{A(x)+ A(y)}{\sqrt 2}}_{\rho_A}
        &\leq \norm{(A(z) \otimes I)\ket{\psi} - L z} + \frac{1}{\sqrt 2} \norm{(A(x) \otimes I)\ket{\psi} - Lx} \\
        &\quad + \frac{1}{\sqrt 2} \norm{(A(y) \otimes I)\ket{\psi} - Ly}
    \end{align*}
    with an analogous bound for Bob.
    By rotational invariance, each of\(x,y,z\) has marginal distribution \(\sigma_n\).
    Squaring the preceding bound, using \((a+b+c)^2 \leq 3a^2 + 3b^2 + 3c^2\), and applying \cref{lem:common-linear-map} therefore gives
    \begin{equation}\label{eq:approx-linearity}
        \E_{x \perp y} \paren*{
        \norm*{A(z) - \frac{A(x) + A(y)}{\sqrt 2}}_{\rho_A}^2 + \norm*{B(z) - \frac{B(x) + B(y)}{\sqrt 2}}_{\rho_B}^2} \leq 12\varepsilon.
    \end{equation}

    We next express the anticommutator errors in terms of consistency and linearity errors.
    The identity
    \begin{multline*}
        \{A(x),A(y)\} \otimes I = (A(x) \otimes B(y) + A(y) \otimes B(x)) \\+ (A(x) \otimes I)(A(y) \otimes I - I \otimes B(y)) + (A(y) \otimes I)(A(x) \otimes I - I \otimes B(x)),
    \end{multline*}
    together with \Cref{eq:seminorm-purification}, the triangle inequality, and the fact that observables have operator norm at most \(1\), gives
    \begin{multline}
        \norm{\{A(x), A(y)\}}_{\rho_A} \leq \norm{(A(x) \otimes B(y) + A(y) \otimes B(x))\ket{\psi}} \\+ \norm{(A(x) \otimes I - I \otimes B(x))\ket{\psi}} + \norm{(A(y) \otimes I - I \otimes B(y))\ket{\psi}} \label{eq:ac-1}.
    \end{multline}
    Similarly, the identity
    \begin{multline*}
    A(x) \otimes B(y) + A(y) \otimes B(x) = 2\paren*{\frac{A(x) + A(y)}{\sqrt 2} \otimes \frac{B(x) + B(y)}{\sqrt 2} - I} \\
    + (A(x) \otimes I)(A(x) \otimes I - I \otimes B(x))+ (A(y) \otimes I)(A(y) \otimes I - I \otimes B(y))
    \end{multline*}
    and the triangle inequality bound the first term on the right-hand side of \Cref{eq:ac-1} by
    \begin{multline}
        \norm{(A(x) \otimes B(y) + A(y) \otimes B(x))\ket{\psi}} \leq 2 \norm*{\paren*{\paren*{\frac{A(x) + A(y)}{\sqrt 2} \otimes \frac{B(x) + B(y)}{\sqrt 2}} - I}\ket{\psi}}\\ + \norm{(A(x) \otimes I - I \otimes B(x))\ket{\psi}}  + \norm{(A(y) \otimes I - I \otimes B(y))\ket{\psi}}\label{eq:ac-2}.
    \end{multline}
    Further, by the triangle inequality we can upper-bound the first term of the right-hand side of \Cref{eq:ac-2} by
    \begin{multline}
        \norm*{\paren*{\paren*{\frac{A(x) + A(y)}{\sqrt 2} \otimes \frac{B(x) + B(y)}{\sqrt 2}} - I}\ket{\psi}} \leq\norm*{\paren*{\paren*{\frac{A(x) + A(y)}{\sqrt 2} - A(z)} \otimes \frac{B(x) + B(y)}{\sqrt 2}}\ket{\psi}} \\
        + \norm*{\paren*{A(z) \otimes \paren*{\frac{B(x) + B(y)}{\sqrt 2} - B(z)}}\ket{\psi}} + \norm{(A(z) \otimes I)(A(z) \otimes I - I \otimes B(z))\ket{\psi}}\label{eq:ac-3}.
    \end{multline}
    The same estimates hold with Alice and Bob interchanged.
    Combining \Cref{eq:ac-1,eq:ac-2,eq:ac-3}, using \linebreak\(\norm*{\frac{P(x)+P(y)}{\sqrt 2}} \leq \sqrt 2\) for \(P \in \{A,B\}\), using Cauchy--Schwarz, and taking expectations, we get
    \[
        \E_{x \perp y} \paren*{ \norm*{ \{A(x),A(y)\}}_{\rho_A}^2 + \norm*{\{B(x),B(y)\}}_{\rho_B}^2 }\leq \bigO(\varepsilon).
    \]
\end{proof}

To apply \cref{thm:GH irrep}, we first establish the required representation-theoretic facts.
In the following let \(n = 2m \geq 2\).
The complex Clifford algebra \(\Cl_n(\C)\) is the unital \(*\)-algebra generated by \(c_1,\dots,c_n\), subject to
\begin{equation}\label{eq:clifford-algebra}
    c_i^*=c_i,\qquad
    c_ic_j+c_jc_i=2\delta_{i,j}.
\end{equation}
For \(x \in S^{n-1}\), define \(c(x) = \sum_{i = 1}^n x_i c_i\). The defining relations imply \(c(x) ^* = c(x), c(x)^2 = 1, \{c(x), c(y)\} = 2\ip{x}{y}\) for all \(x,y \in S^{n-1}\).
Since \(c(e_i) = c_i\), these matrices generate \(\Cl_n(\C)\) as an algebra.
For \(S = \set{i_1, \dots, i_k} \subseteq [n], i_1 < \cdots < i_k\), write \(c_S = c_{i_1} \cdots c_{i_k}\) with \(c_\emptyset = 1\).
The defining relations imply that these ordered monomials span \(\Cl_n(\C)\).

The Majorana operators induce a unital \(*\)-representation \(\pi \colon \Cl_n(\C) \to \operatorname{M}_{2^m}(\C), c_i \mapsto \gamma_i\) where \(\operatorname{M}_{2^m}(\C)\) denotes the algebra of \(2^m\times 2^m\) complex matrices.
By linearity, \(\pi(c(x)) = \gamma(x)\) for every \(x \in S^{n-1}\).
Writing \(\gamma_S = \pi(c_S)\), we have
\begin{equation}\label{eq:majorana-traces}
    \tr[\gamma_S^\dagger \gamma_T] = 2^m\delta_{S,T}, \qquad \tr[\gamma_S] = 0 \text{ for } S \neq \emptyset.
\end{equation}
The matrices \(\gamma_S\) form an orthogonal basis of \(\operatorname{M}_{2^m}(\C)\).
Since the \(c_S\) span \(\Cl_n(\C)\), it follows that \(\pi\) is a \(*\)-isomorphism.
In particular, for $n=2m$, \(\pi\) is the unique irreducible \(*\)-representation of \(\Cl_n(\C)\), up to unitary equivalence.

The multiplicative subgroup generated by \(c_1, \dots, c_n\) is isomorphic to the Clifford group \(C_n\) from \Cref{eq:Cliffordgroup}.
Every element of \(C_n\) has a unique normal form
\[
    J^{a_0} c_1^{a_1} \cdots c_n^{a_n}, \quad a_0, \dots, a_n \in \set{0,1},
\]
and hence \(\abs{C_n} = 2^{n+1}\).

Restricting \(\pi\) to \(C_n\) gives the irreducible unitary representation
\[
    \lambda(J) = -I, \quad \lambda(c_i) = \gamma_i.
\]
Conversely, every unitary representation of \(C_n\) with \(J \mapsto -I\) extends to a \(*\)-representation of \(\Cl_n(\C)\).
By uniqueness of the irreducible representation of the algebra, \(\lambda\) is therefore the unique irreducible representation of \(C_n\) with \(J \mapsto -I\), up to unitary equivalence.
Its dimension is \(d_\lambda = 2^m\), and \cref{eq:majorana-traces} gives its character:
\begin{equation}\label{eq:clifford-character}
    \chi_\lambda(g)=
    \begin{cases}
        d_\lambda,  & g=1,\\
        -d_\lambda, & g=J,\\
        0,          & \text{otherwise}.
    \end{cases}
\end{equation}

For each orthonormal basis~\(Q = (q_1, \dots, q_n)\) of \(\R^n\), define the unitary-valued function \(f_A^Q\) on \(C_n\) by
    \[
        f_A^Q(J^{a_0}c_1^{a_1} \dots c_n^{a_n}) = (-1)^{a_0}A(q_1)^{a_1} \dots A(q_n)^{a_n}
    \]
and define \(f_B^Q\) analogously.
By construction, \(f_A^{Q}(gJ) = -f_A^{Q}(g)\) and \(f_B^Q(gJ) = -f_B^Q(g)\).

The next lemma bounds the average squared error in a product of self-adjoint unitaries by the sum of the individual squared errors.

\begin{lemma}\label{lemma:technical-lemma-0}
    Let \(R_1, \dots, R_k\) be self-adjoint unitaries on a Hilbert space \(\mathcal{H}\) and let \(\ket{\phi} \in \mathcal{H}\) be a state.
    Then
    \[
        \E_{a \in \set{0,1}^k} \norm*{(R_1^{a_1} \cdots R_k^{a_k} - I)\ket{\phi}}^2 \leq \sum_{j = 1}^k \norm{(R_j - I)\ket{\phi}}^2
    \]
\end{lemma}

\begin{proof}
    Set $Q_j=\frac{I+R_j}{2}$, $v_{k+1}=\ket{\phi}$ and $v_j=Q_jv_{j+1}$. Decompose
\[ 
2\ket{\phi}-v_j =Q_j(2\ket{\phi}-v_{j+1})+2(I-Q_j)\ket{\phi} 
\]
into two orthogonal vectors. Since the vectors are orthogonal, it holds $\|Q_j\|\leq 1$ and $2(I-Q_j)=I-R_j$, we get
\begin{align}
 \|2\ket{\phi}-v_j\|^2 &=\|Q_j(2\ket{\phi}-v_{j+1})\|^2+\|2(I-Q_j)\ket{\phi}\|^2\nonumber\\ 
&\le \|2\ket{\phi}-v_{j+1}\|^2+\|(R_j-I)\ket{\phi}\|^2.\label{eq:vecbound}
\end{align}
Summing over $j$ in \Cref{eq:vecbound} and using $\|\ket{\phi}\| = 1$ gives
\begin{align}
\|2\ket{\phi}-v_1\|^2 \le 1+\sum_j\|(R_j-I)\ket{\phi}\|^2.\label{eq:vecbound2} 
\end{align}
It holds $\frac12\sum_{a_j\in\{0,1\}}R_j^{a_j} = \frac{I+R_j}{2}=Q_j$ and therefore
\[ \mathbb E_{a\in\{0,1\}^k} R_1^{a_1}\cdots R_k^{a_k}\ket{\phi}=Q_1\cdots Q_k\ket{\phi}=v_1. \]
Since each $R_1^{a_1}\cdots R_k^{a_k}$ is unitary and by \Cref{eq:vecbound2}, we obtain
\begin{align*}
 \mathbb E_a\|(R_1^{a_1}\cdots R_k^{a_k}-I)\ket{\phi}\|^2 &=2-2\operatorname{Re}\bra{\phi}v_1\rangle\\ 
&=\|2\ket{\phi}-v_1\|^2-1-\|\ket{\phi}-v_1\|^2\\ 
&\le \sum_j\|(R_j-I)\ket{\phi}\|^2.
\end{align*}
\end{proof}

We now combine \cref{lemma:technical-lemma-0} with the approximate Clifford relations to show that, on average over \(Q\) and the group element, \(f_P^Q\) approximately preserves multiplication by each generator \(c_i\).

\begin{lemma}\label{lemma:technical-lemma-1}
    Let \(n \geq 2\) and let \(0 \leq \varepsilon < 1\).
    Consider an \(\varepsilon\)-optimal quantum strategy \( S =(A, B, \ket{\psi})\) for the Grothendieck game~\(G_n\).
    For every \(i \in [n]\), it holds that
    \[
        \E_{Q \in O(n)}\E_{h \in C_n} \paren*{ \norm*{f_A^Q(h)A(q_i) - f_A^Q(hc_i)}_{\rho_A}^2 + \norm*{f_B^Q(h)B(q_i) - f_B^Q(hc_i)}_{\rho_B}^2} \leq \bigO(n\varepsilon)
    \]
    where \(Q = (q_1, \dots, q_n)\) is a Haar-random orthonormal basis of \(\R^n\).
\end{lemma}

\begin{proof}
    For \( i = n\), the expectation vanishes.
    We therefore assume \(i < n\).
    The normal forms of \(h\) and \(hc_i\) are
    \begin{align*}
        h = J^{a_0}c_1^{a_1} \cdots c_n^{a_n}, \quad hc_i = J^{a_0 + \sum_{j = i+1}^n a_j} c_1^{a_1} \cdots c_{i-1}^{a_{i-1}}c_i^{1 \oplus a_i} c_{i+1}^{a_{i+1}} \cdots c_n^{a_n}
    \end{align*}
    and writing \(U_{P,i}\coloneq P(q_i)^{a_i} \cdots P(q_n)^{a_n}\) and \(U_{P,n+1} = I\) for \(P \in \set{A,B}\), we obtain the following estimates using \Cref{eq:seminorm-purification}, the isometry invariance in \Cref{lem:state-dependent-seminorm} and the triangle inequality:
    \begin{align*}
        &\norm*{f_A^Q(h)A(q_i) - f_A^Q(hc_i)}_{\rho_A}^2 \\
        &= \norm{(U_{A,i+1}A(q_i) - (-1)^{\sum_{j = i+1}^n a_j} A(q_i) U_{A,i+1}) \otimes I\ket\psi}^2 \\
        &=\norm{(U_{A, i+1}\otimes U_{B, i+1})(A(q_i) \otimes I)\ket\psi - (-1)^{\sum_{j = i+1}^n a_j} (A(q_i) \otimes I)(U_{A,i+1} \otimes U_{B, i+1})\ket\psi}^2 \\
        &\leq 2 \norm{(U_{A,i+1} \otimes U_{B,i+1} - (-1)^{\sum_{j = i+1}^n a_j}I)(A(q_i) \otimes I)\ket\psi}^2 + 2\norm{(U_{A,i+1} \otimes U_{B,i+1} - I)\ket\psi}^2
    \end{align*}
    The expression depend only on \(a_{i+1}, \dots, a_n\).
    Thus averaging over \(h \in C_n\) is equivalent to averaging over these bits.
    Applying \cref{lemma:technical-lemma-0} to the observables \((-A(q_{j}) \otimes B(q_{j}))_{j = i+1}^n\) and the state \((A(q_i) \otimes I)\ket\psi\) as well as to the observables \((A(q_j) \otimes B(q_j))_{j = i+1}^n\) and the state \(\ket\psi\) gives
    \begin{align*}
        &\E_{h \in C_n} \norm{f_A^Q(h)A(q_i) - f_A^Q(hc_i)}_{\rho_A}^2 \\&\quad \leq 2\sum_{j = i+1}^n \norm{(A(q_j) \otimes B(q_j) - I)\ket\psi}^2 + 2\sum_{j = i+1}^n \norm{(A(q_j) \otimes B(q_j) + I)(A(q_i) \otimes I)\ket{\psi}}^2
    \end{align*}
    Using the identities
    \begin{align*}
        A(q_j) \otimes B(q_j) - I &= (A(q_j) \otimes I)(I \otimes B(q_j) - A(q_j) \otimes I) \\
        (A(q_j) \otimes B(q_j) + I)(A(q_i) \otimes I) &= (I \otimes B(q_j))\paren*{ \{A(q_j), A(q_i)\} \otimes I - (A(q_i) \otimes I)(A(q_j) \otimes I - I \otimes B(q_j))}
    \end{align*}
    and applying the triangle inequality and unitary invariance, we obtain
    \[
        \E_{h \in C_n} \norm*{f_A^Q(h) A(q_i) - f_A^Q(hc_i)}_{\rho_A}^2 \leq 4 \sum_{j = i+1}^n \norm{\{A(q_i), A(q_j)\}}_{\rho_A}^2 + 6 \sum_{j = i+1}^n \norm{(A(q_j) \otimes I - I \otimes B(q_j))\ket{\psi}}^2
    \]
    and the same argument for Bob gives
    \[
        \E_{h \in C_n} \norm*{f_B^Q(h) B(q_i) - f_B^Q(hc_i)}_{\rho_B}^2 \leq 4 \sum_{j = i+1}^n \norm{\{B(q_i), B(q_j)\}}_{\rho_B}^2 + 6 \sum_{j = i+1}^n \norm{(A(q_j) \otimes I - I \otimes B(q_j))\ket{\psi}}^2
    \]
    For Haar-random \(Q\), each \(q_j\) is uniform on the sphere, and each pair \((q_i, q_j), i \neq j\) is a uniform orthonormal pair.
    Taking the expectation over \(Q\), adding the two players' bounds, and applying \cref{lemma:average-clifford-relations} gives
    \[
        \E_{Q \in O(n)}\E_{h \in C_n} \paren*{ \norm{f_A^Q(h)A(q_i) - f_A^Q(hc_i)}_{\rho_A}^2 + \norm{f_B^Q(h)B(q_i) - f_B^Q(hc_i)}_{\rho_B}^2 }\leq \bigO((n-i)\varepsilon).
    \]
\end{proof}

For an orthonormal basis \(Q = (q_1, \dots, q_n)\) of \(\R^n\), define
\[
    \gamma_Q(x) \coloneq \sum_{i = 1}^n \ip{q_i}{x}\gamma_i.
\]
Approximate consistency of the ideal observables also determines the shared state.
The next lemma shows that it must be close to a maximally entangled state (tensored with an auxiliary state).

\begin{lemma}\label{lemma:technical-lemma-2}
    Let \(n = 2m\) with \(m \geq 1\), let \(\ket{\phi} \in \C^{2^m} \otimes \C^{2^m} \otimes \mathcal{H}_{\mathrm{aux}}\) be an arbitrary state, and let \(\hat Q\) be any orthonormal basis of \(\R^n\).
    Suppose that
    \[
        \E_{x \sim\sigma_n}\norm{(\gamma_{\hat Q}(x) \otimes I \otimes I_{\mathrm{aux}}- I \otimes \overline{\gamma_{\hat Q}(x)} \otimes I_\mathrm{aux})\ket{\phi}}^2 \leq \delta,
    \]
    for some \(\delta \geq 0\).
    Then there exists an auxiliary state~\(\ket{\mathrm{aux}} \in\mathcal{H}_\mathrm{aux}\) such that
    \[
        \norm{\ket\phi - \ket{\psi_{2^m}^+} \otimes \ket{\mathrm{aux}}} \leq \sqrt{\frac{n\delta}2}.
    \]
\end{lemma}

\begin{proof}
    Let \(D \coloneq \sum_{i = 1}^n \paren*{\gamma_i \otimes I \otimes I_{\mathrm{aux}}- I \otimes \overline{\gamma_i} \otimes I_\mathrm{aux}}^2\) and \(P\) be the projection onto \( \operatorname{span}\set*{\ket{\psi^+_{2^m}}} \otimes \mathcal{H}_\mathrm{aux}\).
    Since \(\gamma_i^2 = I\) it holds that
    \[
        D = 2nI - 2\sum_i \gamma_i \otimes \overline{\gamma_i} \otimes I_{\mathrm{aux}}.
    \]
    By \Cref{eq:majorana-traces}, the products \(\gamma_S\), \(S \subseteq [n]\), form an orthogonal basis of \(\operatorname{M}_{2^m}(\C)\).

    By identifying matrices with vectors through the isometry
    \[
        \operatorname{vec}\colon \operatorname{M}_{2^m}(\C) \to \C^{2^m} \otimes \C^{2^m}, \quad X \mapsto \sum_{i,j} X_{i,j} \ket{i}\otimes \ket{j},
    \]
    under which \((\gamma_i \otimes \overline{\gamma_i})\operatorname{vec}(X) = \operatorname{vec}(\gamma_i X \gamma_i)\).
    The anticommutation relations of the Majorana operators give, for \(\abs{S} = k\),
    \[
        \sum_{i=1}^n \gamma_i \gamma_S \gamma_i = \sum_{i \notin S} \gamma_i \gamma_S \gamma_i + \sum_{i \in S} \gamma_i \gamma_S \gamma_i  = (-1)^k(n-k)\gamma_S + (-1)^{k-1}k\gamma_S = (-1)^k(n-2k)\gamma_S
    \]
    so \(D\) has eigenvalues
    \[
        p_k = 2n - 2(-1)^k(n-2k).
    \]
    That is, for even \(k\) we have \(p_k = 4k\) and for odd \(k\) we find \(p_k = 4(n - k)\).
    Hence \(p_k = 0\) precisely when \(k = 0\), and every nonzero eigenvalue is at least \(4\).
    Since \(\operatorname{vec}(I) = 2^{m/2}\ket{\psi_{2^m}^+}\), the kernel of \(D\) is the range of \(P\).
    Therefore,
    \[
        D \geq 4(I-P).
    \]

    Apply \Cref{eq:seminorm-sphere} with \(\rho = \ket\phi\bra\phi\), using the orthonormal coordinates determined by \(\hat Q\).
    Together with the hypothesis, this gives
    \begin{align*}
        \sum_{i = 1}^n \norm{(\gamma_i \otimes I \otimes I_\mathrm{aux} - I \otimes \overline{\gamma_i} \otimes I_\mathrm{aux})\ket{\phi}}^2 &= n \E_{x \sim\sigma_n} \norm{(\gamma_{\hat Q}(x) \otimes I \otimes I_\mathrm{aux} - I \otimes \overline{\gamma_{\hat Q}(x)} \otimes I_\mathrm{aux})\ket\phi}^2 \\
        &\leq \delta n
    \end{align*}
    Combining this estimate with \(D \geq 4(I - P)\) gives
    \[
        \norm{(I - P)\ket{\phi}}^2 \leq \frac 14 \bra{\phi}D\ket{\phi} \leq \frac{n\delta}4.
    \]

    Let \(r = \norm{(\bra{\psi_{2^m}^+} \otimes I_\mathrm{aux})\ket{\phi}}\).
    If \(r > 0\), define
    \[
        \ket{\mathrm{aux}} \coloneq \frac{(\bra{\psi_{2^m}^+} \otimes I_{\mathrm{aux}})\ket{\phi}}r,
    \]
    and if \(r = 0\), choose any auxiliary state.
    In either case,
    \begin{align*}
        \norm{\ket{\phi} - \ket{\psi_{2^m}^+} \otimes\ket{\mathrm{aux}}}^2 = 2 - 2 \operatorname{Re} \ip{\phi}{\psi_{2^m}^+ \otimes \mathrm{aux}}
        = 2(1-r)
        \leq 2 (1-r^2) = 2 \norm{(I - P)\ket{\phi}}^2 \leq \frac{n\delta}2.
    \end{align*}
    Taking square roots completes the proof.
\end{proof}

We now combine the linear approximation, approximate multiplicativity, and state-extraction estimates to prove robust self-testing of the Majorana observables.

\begin{theorem}\label{thm:self-testing}
    Let \(m \geq 1\) be an integer, \(0 \leq \varepsilon < 1\), and set \(n = 2m\).
    Consider an \(\varepsilon\)-optimal quantum strategy \(S = (A,B,\ket\psi) \) for the Grothendieck game~\(G_n\).
    Then there exist local isometries \( \Phi_A \colon \mathcal{H}_A \to (\C^2)^{\otimes m} \otimes \mathcal{H}_{A, \mathrm{aux}}\), \(\Phi_B \colon \mathcal{H}_B \to (\C^2)^{\otimes m} \otimes \mathcal{H}_{B, \mathrm{aux}}\), and a state \(\ket{\mathrm{aux}} \in \mathcal{H}_{A, \mathrm{aux}} \otimes \mathcal{H}_{B, \mathrm{aux}}\) such that
    \begin{align*}
        \norm[\big]{(\Phi_A \otimes \Phi_B)\ket{\psi} - \ket{\psi^+_{2^m}} \otimes \ket{\mathrm{aux}}}^2 &\leq \bigO(n^2\varepsilon), \\
        \E_{x \sim \sigma_n} \norm[\big]{ \Phi_A A(x) - (\gamma(x) \otimes I)\Phi_A}_{\rho_A}^2 &\leq \bigO(n\varepsilon), \\
        \E_{y \sim \sigma_n} \norm[\big]{\Phi_B B(y) - (\overline{\gamma(y)} \otimes I)\Phi_B}_{\rho_B}^2 &\leq \bigO(n\varepsilon).
    \end{align*}
\end{theorem}

\begin{proof}
    We first choose an orthonormal basis for which both the approximate multiplicativity and linear approximation errors are small.
    Averaging \cref{lemma:technical-lemma-1} over \(i \in [n]\) and combining the result \cref{lemma:technical-lemma-3} gives
    \[
        \E_{Q \in O(n)} \sum_{P \in \set{A,B}} \paren*{n^{-1} \E_{h \in C_n} \E_{i \in [n]} \norm{f_P^Q(h)P(q_i) - f_P^Q(hc_i)}_{\rho_P}^2 + \E_{x \sim\sigma_n}\norm{P(x) - \tilde P_Q(x)}_{\rho_P}^2} \leq \bigO(\varepsilon).
    \]
    Thus, there exists a basis~\(\hat Q = (\hat q_1, \dots, \hat q_n)\) such that
    \[
        \E_{h \in C_n}\E_{i \in [n]}\paren*{ \norm{f_A^{\hat Q}(h)A({\hat q_i}) - f_A^{\hat Q}(hc_i)}_{\rho_A}^2 + \norm{f_B^{\hat Q}(h)B({\hat q_i}) - f_B^{\hat Q}(hc_i)}_{\rho_B}^2} \leq \bigO(n\varepsilon)
    \]
    and simultaneously
    \begin{equation}\label{eq:lin-approx-operator}
        \E_{x \sim \sigma_n} \paren*{\norm{A(x) - \tilde A_{\hat Q}(x)}_{\rho_A}^2 + \norm{B(x) - \tilde B_{\hat Q}(x)}_{\rho_B}^2}\leq \bigO(\varepsilon).
    \end{equation}

    We next apply the Gowers--Hatami theorem to extract the Majorana observables on this basis.
    By \Cref{eq:clifford-character} and \(f_P^{\hat Q}(gJ) = -f_P^{\hat Q}(g)\), the character condition in \cref{thm:GH irrep} holds:
    \[
        d_\lambda \E_{h \in C_n} \overline{\chi_\lambda(h)}f_P^{\hat Q}(gh) = \frac{d_\lambda^2}{\abs{C_n}}(f_P^{\hat Q}(g) - f_P^{\hat Q}(gJ)) = f_P^{\hat Q}(g), \quad P\in \{A,B\}.
    \]
    Let \(\lambda^\star\) denote the dual representation, given in the chosen dual basis by
    \[
        \lambda^\star(g) = \lambda(g^{-1})^T = \overline{\lambda(g)}.
    \]
    Since \(c_i^{-1} = c_i\), we have \(\lambda^\star(c_i) = \lambda(c_i)^T = \overline{\gamma_i}\).
    Moreover, \(\chi_{\lambda^\star} = \overline{\chi_\lambda} = \chi_\lambda\), so the same character condition holds for \(f_B^{\hat Q}\) with respect to \(\lambda^\star\).
    Applying \cref{thm:GH irrep} to \(f_A^{\hat Q}\) with representation \(\lambda\), and to \(f_B^{\hat Q}\) with representation \(\lambda^\star\) gives
    \[
        \E_{i \in [n]}\paren*{\norm{(\gamma_i\otimes I_{A,\mathrm{aux}})V_\lambda - V_\lambda A({\hat q_i})}_{\rho_A}^2 + \norm{(\overline{\gamma_i} \otimes I_{B,\mathrm{aux}})V_{\lambda^\star} - V_{\lambda^\star} B({\hat q_i})}_{\rho_B}^2} \leq \bigO(n\varepsilon)
    \]
        for isometries \(V_\lambda \colon \mathcal{H}_A \to \C^{2^m} \otimes \mathcal{H}_{A,\mathrm{aux}}\), \(V_{\lambda^\star} \colon \mathcal{H}_B \to (\C^{2^m})^* \otimes \mathcal{H}_{B, \mathrm{aux}}\) where \(\mathcal{H}_{A, \mathrm{aux}} = (\C^{2^m})^* \otimes \mathcal{H}_A\), \( \mathcal{H}_{B, \mathrm{aux}} = \C^{2^m} \otimes \mathcal{H}_B\).
        We identify \( (\C^{2^m})^*\) with \(\C^{2^m}\) using the dual of the fixed standard basis.

    We extend these estimates from the basis vectors to the whole sphere.
    By \Cref{eq:seminorm-sphere} we have
    \[
        \E_{x\sim \sigma_n} \norm{(\gamma_{\hat Q}(x) \otimes I_{A,\mathrm{aux}})V_\lambda - V_\lambda\tilde A_{\hat Q}(x)}_{\rho_A}^2 = \frac 1n \sum_i \norm{(\gamma_i\otimes I_{A,\mathrm{aux}})V_\lambda - V_\lambda A({\hat q_i})}_{\rho_A}^2 \leq \bigO(n\varepsilon).
    \]
    Combining this estimate and its Bob analogue with \Cref{eq:lin-approx-operator}, the triangle inequality, and the isometry invariance in \Cref{lem:state-dependent-seminorm}, we obtain
    \begin{equation}
        \E_{x\sim \sigma_n} \paren*{\norm*{(\gamma_{\hat Q}(x) \otimes I_{A, \mathrm{aux}})V_\lambda - V_\lambda A(x)}_{\rho_A}^2 + \norm*{(\overline{\gamma_{\hat Q}(x)} \otimes I_{B, \mathrm{aux}})V_{\lambda^\star} - V_{\lambda^\star} B(x)}_{\rho_B}^2} \leq \bigO(n\varepsilon) \label{eq:something}.
    \end{equation}
    Regroup the output factors as \(\C^{2^m} \otimes \C^{2^m} \otimes \mathcal H_\mathrm{aux}\) where \(\mathcal H_\mathrm{aux} = \mathcal H_{A,\mathrm{aux}} \otimes \mathcal H_{B,\mathrm{aux}}\).
    By the triangle inequality, \Cref{eq:seminorm-purification}, and isometry invariance (\Cref{lem:state-dependent-seminorm}), we get
    \begin{align*}
        &\E_{x \sim\sigma_n} \norm{(\gamma_{\hat Q}(x) \otimes I \otimes I_\mathrm{aux}- I \otimes \overline{\gamma_{\hat Q}(x)} \otimes I_\mathrm{aux})(V_\lambda \otimes V_{\lambda^\star})\ket\psi}^2 \\
        &\leq 3 \E_{x \sim\sigma_n} \norm{(((\gamma_{\hat Q}(x) \otimes I_{A, \mathrm{aux}})V_\lambda - V_\lambda A(x)) \otimes V_{\lambda^\star})\ket\psi}^2 \\
        &\quad + 3 \E_{x\sim\sigma_n} \norm{(V_\lambda \otimes V_{\lambda^\star})(A(x) \otimes I - I \otimes B(x))\ket\psi}^2 \\
        &\quad + 3 \E_{x \sim\sigma_n} \norm{(V_\lambda \otimes (V_{\lambda^\star} B(x) - (\overline{\gamma_{\hat Q}(x)} \otimes I_{B, \mathrm{aux}})V_{\lambda^\star}))\ket\psi}^2 \\
        &\leq \bigO(n\varepsilon)
    \end{align*}
    where the last inequality follows from \Cref{eq:consistency,eq:something}.

    Applying \cref{lemma:technical-lemma-2} to \(\ket{\phi} = (V_\lambda \otimes V_{\lambda^\star})\ket{\psi}\) with \(\delta = \bigO(n\varepsilon)\), we get
    \begin{equation}
        \norm{(V_\lambda \otimes V_{\lambda^\star})\ket\psi - \ket{\psi_{2^m}^+} \otimes \ket{\mathrm{aux}}}^2 \leq \bigO(n^2\varepsilon) \label{eq:state-close}
    \end{equation}

Since \(\hat Q\) is an orthonormal basis, the matrices \(\gamma(\hat q_i)\) satisfy the same relations as \(\gamma_i\) and generate the same matrix algebra.
They therefore define an irreducible representation of \(C_n\) in which \(J\) acts as \(-I\).
By uniqueness of this representation, there exists a unitary \(U_{\hat Q}\) such that
\begin{equation}\label{eq:majorana-change-basis}
        U_{\hat Q} \gamma_i U_{\hat Q}^\dagger = \gamma(\hat q_i), \quad i \in [n].
\end{equation}
Hence,
\begin{align*}
    U_{\hat Q}\gamma_{\hat Q}(x)U_{\hat Q}^\dagger &= \sum_i \ip{\hat q_i}{x}\gamma(\hat q_i) = \gamma(x)
\end{align*}
The maximally entangled state in invariant under \(U_{\hat Q} \otimes \overline{U_{\hat Q}}\).
Define \(\Phi_A = (U_{\hat Q} \otimes I_{A, \mathrm{aux}})V_\lambda\), \(\Phi_B = (\overline{U_{\hat Q}} \otimes I_{B, \mathrm{aux}})V_{\lambda^\star}\).
Isometry invariance (\Cref{eq:state-dependent-seminorm}) and \Cref{eq:something} give
\[
    \E_{x \sim \sigma_n} \paren*{ \norm*{ \Phi_A A(x) - (\gamma(x) \otimes I)\Phi_A}_{\rho_A}^2 + \norm*{ \Phi_B B(x) - (\overline{\gamma(x)} \otimes I)\Phi_B}_{\rho_B}^2} \leq \bigO(n\varepsilon).
\]
This proves the observable bounds.
The state bound follows from \Cref{eq:state-close} and the invariance of \(\ket{\psi_{2^m}^+}\) under \(U_{\hat Q}\otimes \overline{U_{\hat Q}}\).
\end{proof}

\begin{remark}
Among bounds that hold uniformly over all quantum strategies and are linear in \(\varepsilon\), the quadratic dependence on \(n\) in the state bound and the linear dependence on \(n\) in each observable bound in \cref{thm:self-testing} are optimal.
\end{remark}
\begin{proof}
Let \(n = 2m \geq 4\).
We use the projection strategy \(a = b\) from \Cref{eq:projection-strategy}, regarded as taking values in \(\R^{n-2}\).
The construction in the proof of \cref{thm:tsirelson} realizes this strategy on \(\ket{\psi_{2^{m-1}}^+}\), with observables
\[
    A(x) = \sum_{j = 1}^{n-2} a_j(x) \tilde \gamma_j, \quad B(y) = \overline{A(y)}
\]
where \(\tilde \gamma_1, \dots, \tilde \gamma_{n-2}\) are the Majorana operators on \(\C^{2^{m-1}}\).
This realization preserves the correlations and hence the bias.
By \Cref{eq:projection-bias}, the resulting quantum strategy is \((n-1)^{-2}\)-optimal.
For any local isometries \(\Phi_A, \Phi_B\) and auxiliary state \(\ket{\mathrm{aux}}\), its state error satisfies
\[
    \norm*{(\Phi_A \otimes \Phi_B)\ket{\psi_{2^{m-1}}^+} - \ket{\psi_{2^m}^+} \otimes \ket{\mathrm{aux}}}^2 = 2 - 2\operatorname{Re} \bra{\psi_{2^{m-1}}^+} (\Phi_A^\dagger \otimes \Phi_B^\dagger)(\ket{\psi_{2^m}^+} \otimes \ket{\mathrm{aux}}) \geq 2 - \sqrt 2
\]
because local isometries preserve the Schmidt rank \(2^{m-1}\) of the initial state, whereas every squared Schmidt coefficient of \(\ket{\psi_{2^m}^+} \otimes \ket{\mathrm{aux}}\) is at most \(2^{-m}\).
Their absolute overlap is therefore at most \(1 / \sqrt 2\).
Hence, a uniform state bound of the form \(f(n)\varepsilon\) must satisfy, for this strategy,
\[
    2- \sqrt 2 \leq \frac{f(n)}{(n-1)^2}.
\]
It follows that \(f(n) = \Omega(n^2)\) along even dimensions.

For the observable bounds, the same strategy has reduced state \(\rho_A = I/2^{m-1}\).
For any local isometry~\(\Phi_A\), \Cref{eq:optimality-self-testing} gives
\begin{align*}
    \E_{x \sim \sigma_n} \norm*{\Phi_A A(x) - (\gamma(x) \otimes I)\Phi_A}_{\rho_A}^2 &= 2 - \frac{2}{2^{m-1}(n-1)}\sum_{j = 1}^{n-2} \operatorname{Re}\tr\left[ \tilde \gamma_j \Phi_A^\dagger (\gamma_j \otimes I)\Phi_A \right] \\
    &\geq 2 - \frac{2(n-2)}{n-1} = \frac{2}{n-1}
\end{align*}
where each trace has absolute value at most \(2^{m-1}\) by the Cauchy--Schwarz inequality and \(\norm{\Phi_A \tilde \gamma_j}_{\mathrm{HS}} = \norm{(\gamma_j \otimes I)\Phi_A}_{\mathrm{HS}} = \sqrt {2^{m-1}}\).
The same argument applies to Bob.
Since this strategy is \( (n-1)^{-2}\)-optimal, any uniform bound \(f(n)\varepsilon\) for the observable estimates must satisfy
\[
    \frac{2}{n-1} \leq \frac{f(n)}{(n-1)^2}
\]
and hence \(f(n) = \Omega(n)\).
\end{proof}

The averaged observable bounds in \cref{thm:self-testing} also yield simultanoues bounds on all measurements indexed by a random orthonormal basis.
For each fixed strategy, a Haar-random basis \(Q = (q_1, \dots, q_n)\) certifies all \(2n\) observables \(A(q_i)\) and \(B(q_i)\) with high probability.
A change of basis on the extracted systems expresses this guarantee in terms of the fixed Majorana operators \(\gamma_i\) and \(\overline{\gamma_i}\).

\begin{corollary}\label{cor:HighProb}
    Let \(n = 2m \geq 2\) be even, let \(0 \leq \varepsilon< 1\), and let \( 0 < \delta < 1\).
    Fix an \(\varepsilon\)-optimal quantum strategy \(S = (A,B,\ket\psi)\) for the Grothendieck game \(G_n\), and let \(Q = (q_1, \dots, q_n)\) be a Haar-random orthonormal basis of \(\R^n\).
    Then with probability at least \(1-\delta\), there exist local isometries \(\Psi_A^Q \colon \mathcal{H}_A \to (\C^2)^{\otimes m} \otimes \mathcal H_{A,\mathrm{aux}}, \Psi_B^Q \colon \mathcal{H}_B \to (\C^2)^{\otimes m} \otimes \mathcal H_{B, \mathrm{aux}}\), and a state \(\ket{\mathrm{aux}}\in \mathcal{H}_{A,\mathrm{aux}} \otimes \mathcal{H}_{B, \mathrm{aux}} \) such that
    \begin{align*}
        \norm*{ (\Psi_A^Q \otimes \Psi_B^Q)\ket\psi - \ket{\psi_{2^m}^+}\otimes \ket{\mathrm{aux}}}^2 &\leq \bigO(n^2\varepsilon)
    \end{align*}
    and, for every \(i \in [n]\),
    \begin{align*}
        \norm*{\Psi_A^Q A(q_i) - (\gamma_i \otimes I)\Psi_A^Q}_{\rho_A}^2 &\leq \bigO(n^2\varepsilon/\delta), \\
        \norm*{\Psi_B^Q B(q_i) - (\overline{\gamma_i} \otimes I)\Psi_B^Q}_{\rho_B}^2 &\leq \bigO(n^2\varepsilon/\delta).
    \end{align*}
\end{corollary}
\begin{proof}
    Apply \cref{thm:self-testing} and fix the resulting isometries \(\Phi_A,\Phi_B\) and auxiliary state.
    Since each \(q_i\) is uniformly distributed on \(S^{n-1}\), there is a universal constant \(C\) such that, for every \(i \in [n]\), we have
    \begin{align*}
    \E_{Q\in O(n)} \norm*{ \Phi_A A(q_i) - (\gamma(q_i) \otimes I)\Phi_A}_{\rho_A}^2 \leq Cn\varepsilon
\end{align*}
for some constant \(C\) and the same estimate holds for Bob.
For \(\varepsilon > 0\), Markov's inequality gives
\[
    \Pr_{Q\in O(n)}\left[ \norm*{\Phi_A A(q_i) -(\gamma(q_i) \otimes I)\Phi_A}_{\rho_A}^2 > \frac{2Cn^2 \varepsilon}{\delta}\right] \leq \frac{\delta}{2n} .
\]
The same estimate holds for Bob.
Taking a union bound over both players and all \(i \in [n]\), we conclude that, with probability at least \(1-\delta\), all \(2n\) squared norms are simultaneously bounded by \(2Cn^2\varepsilon/\delta\).
For \(\varepsilon = 0\), all these squared norms vanish almost surely.

As in \Cref{eq:majorana-change-basis}, choose a unitary \(U_Q\) satisfying \(U_Q \gamma_i U_Q^\dagger = \gamma(q_i)\).
Define \(\Psi_A^Q = (U_Q^\dagger \otimes I)\Phi_A,  \Psi_B^Q = (U_Q^T \otimes I)\Phi_B.\)
This change of basis replaces \(\gamma(q_i)\) and \(\overline{\gamma(q_i)}\) by \(\gamma_i\) and \(\overline{\gamma_i}\), respectively.
The observable bounds follow from isometry invariance (\Cref{lem:state-dependent-seminorm}).
The state bound is unchanged because \((U_Q^\dagger \otimes U_Q^T)\ket{\psi_{2^m}^+} = \ket{\psi_{2^m}^+}\).
\end{proof}

\subsection{Robust Self-Test for the Spin Representation}
\label{sub:robust_self_test_for_the_spin_representation}

In \cref{thm:self-testing} we proved that any near-optimal strategy for the Grothendieck game is close (up to local isometries) to the ideal strategy given by the Majorana observables~$\gamma(x)$ (acting on half of a maximally entangled state).
Recall from \cref{sec:robust-majorana-self-test} that the Clifford algebra is generated (as an algebra) by self-adjoint unitaries~$c(x)$ for~$x \in S^{n-1}$ subject to the relation $\{c(x),c(y)\}=2\langle x, y \rangle$; the assignment $c(x) \mapsto \gamma(x)$ defines an irreducible representation (unique up to isomorphism) if $n=2m$, which we assume throughout this section.
Here we show that the self-testing result extends to the pin group $\Pin(n)$, which we recall is generated by the elements $c(x)$, and its representation
\begin{align*}
    \Gamma \colon \Pin(n) \to \U(\C^{2^m}), \quad c(x_1) \cdots c(x_k) \mapsto \gamma(x_1) \cdots \gamma(x_k) \quad (k\in\N, \, x_1,\dots,x_k \in S^{n-1}),
\end{align*}
which is known as the \emph{spin representation}; it is irreducible.
The pin group is a double cover of the orthogonal group~$\Or(n)$, for which the above is a projective representation: any element in $\Or(n)$ lifts to two elements~$\pm g$ such that $\Gamma(-g)=-\Gamma(g)$.

Now, any element of the pin group can be written as a signed product of no more than~$n$ generators~$c(x)$:
for every $g \in \Pin(n)$, there exist $s\in\{\pm1\}$, $0 \leq k \leq n$, and $x_j \in S^{n-1}$ such that
\begin{align}\label{eq:pin decomposition}
    g = s \, c(x_1) \cdots c(x_k).
\end{align}
The main challenge in lifting the self-testing result from the generators to the group is that such a decomposition is \emph{not} unique, so it is not a priori clear how the operators $A(x)$ can be extended to an approximate representation of the pin group.
However, given a group element~$g$, a decomposition~\eqref{eq:pin decomposition} can certainly be computed in a measurable way (for example, by an algorithm decomposing the corresponding orthogonal transformation  into a product of Householder reflections, lifting to the pin group, and choosing the appropriate sign).
Thus, there exist measurable functions $s=s(g)$, $k=k(g)$, $x_j=x_j(g)$ such that \cref{eq:pin decomposition} holds.

Now, a uniformly (Haar) random choice of $g$ need \emph{not} give rise to uniformly random choices of the unit vectors~$x_j$, so \cref{thm:self-testing} still cannot be applied directly, but we can use a symmetrization trick to overcome this issue.
To this end, we consider for a given rotation $R \in \SO(n)$ the automorphism $\alpha_R \colon \Pin(n) \to \Pin(n)$ defined by $\alpha_R(c(x)) \coloneqq c(Rx)$ for $x\in S^{n-1}$;
this is well-defined because $\{c(Rx),c(Ry)\} = \{c(x),c(y)\}$.
Then, instead of decomposing $g$, we decompose $\alpha_R^{-1}(g) = s c(x_1) \cdots c(x_k)$.
Then, $g = s c(R x_1) \cdots c(R x_k)$.
Crucially, for independent Haar random choices of $g$ and $R$, all vectors $R x_j$ will be uniformly distributed; while they need not be independent for different $j$, this will still be sufficient to apply \cref{thm:self-testing} and prove our self-testing result.
Formally, given Alice's observables $A(x)$ in a quantum strategy, we define
\begin{align*}
    A_R \colon \Pin(n) \to \U(\mathcal H_A), \quad A_R(g) \coloneq s(h) A(R x_1(h)) \cdots A(R x_{k(h)}(h)) \qquad (h \coloneqq \alpha_R^{-1}(g)).
\end{align*}
We similarly define $B_R$ using Bob's observables $B(y)$.
Crucially, these are implementable by Alice and Bob if they have unitary access to their observables.
We now show that for a near-optimal strategy for the Grothendieck game, the operators $A_R(g)$ and $B_R(g)$ approximately implement the spin representation.

\begin{theorem}[Spin representation self-test]\label{thm:spinselftest}
In the situation of \cref{thm:self-testing}, we also have (for the same local isometries and state):
    \begin{align*}
        \E_{R \in \SO(n)} \E_{g \in \Pin(n)} \norm*{\Phi_A A_R(g) - (\Gamma(g) \otimes I)\Phi_A}_{\rho_A}^2 &\leq \bigO(n^3 \varepsilon), \\
        \E_{R \in \SO(n)} \E_{g \in \Pin(n)} \norm*{\Phi_B B_R(g) - (\overline{\Gamma(g)}\otimes I)\Phi_B}_{\rho_B}^2 &\leq \bigO(n^3 \varepsilon).
    \end{align*}
\end{theorem}
\begin{proof}
We only prove the first inequality; the second one is completely analogous.
First note that the second estimate in \cref{thm:self-testing} reads:
\begin{align}\label{eq:bound 1}
    \E_{x \sim \sigma_n} \norm*{\Phi_A A(x) - (\gamma(x) \ot I)\Phi_A}_{\rho_A}^2 &\leq \bigO(n \varepsilon).
\end{align}
Furthermore, we also know from \cref{lemma:average-clifford-relations} that
\begin{align}\label{eq:bound 2}
    \E_{x \sim \sigma_n} \norm{(A(x) \otimes I - I \otimes B(x))\ket{\psi}}^2 &\leq \bigO(\varepsilon)
\end{align}
Now, for any fixed choice of unit vectors $x_1,\dots,x_k \in S^{n-1}$, where $k \leq n$, we have
\begin{align*}
&\quad \norm*{\parens*{ \Phi_A A(x_1) \cdots A(x_k) \ot I - (\gamma(x_1) \cdots \gamma(x_k) \ot I) \Phi_A \otimes I } \ket\psi} \\
&\leq \norm*{\parens*{ \Phi_A A(x_1) \ot B(x_k) \cdots B(x_2) - (\gamma(x_1) \cdots \gamma(x_k) \ot I) \Phi_A \otimes I }\ket{\psi}}
+ \sum_{\ell=2}^k \norm*{\left( A(x_\ell) \ot I - I \ot B(x_\ell) \right) \ket\psi} \\
&\leq \sum_{\ell=1}^k \norm*{\left( \Phi_A A(x_\ell) \ot I - (\gamma(x_\ell) \ot I) \Phi_A \ot I  \right) \ket\psi}
+ 2 \sum_{\ell=2}^k \norm*{\left(A(x_\ell) \ot I - I \ot B(x_\ell) \right) \ket\psi}
\end{align*}
by telescoping sums (first to swap all but the first $A(x_\ell)$ to $B(x_\ell)$, which reverses the order; then we successively replace $\Phi_A A(x_1)$ by $\gamma(x_1) \Phi_A$, swap $B(x_2)$ back to $A(x_2)$, and so forth until we are done).
Hence, using the Cauchy--Schwarz inequality and the linearity of expectation values,
\begin{align*}
&\quad \E_{R \in \SO(n)} \norm*{\Phi_A A(R x_1) \cdots A(R x_k) - \left(\gamma(R x_1) \cdots \gamma(R x_k) \otimes I\right) \Phi_A}_{\rho_A}^2 \\
&\leq 5k \E_{R \in \SO(n)}
\left(
\sum_{\ell=1}^k  \norm*{\Phi_A A(R x_\ell)- \left(\gamma(R x_\ell) \otimes I\right)\Phi_A}_{\rho_A}^2 + \sum_{\ell=2}^k  \norm*{\left(A(R x_\ell) \ot I - I \ot B(R x_\ell) \right) \ket\psi}^2 \right) \\
&\leq \bigO(n^3 \eps).
\end{align*}
In the last step we used that each $R x_\ell$ is uniformly random on the sphere~$S^{n-1}$, so we can apply \cref{eq:bound 1,eq:bound 2}.

Finally, if $g \in \Pin(n)$ is Haar random then so is $h = \alpha_R^{-1}(g)$ for any fixed~$R$.
Thus, $h$ and $R$ are also independent.
Moreover, if $h = s c(x_1) \cdots c(x_k)$ then $g = s c(R x_1) \cdots c(R x_k)$.
As a consequence:
\begin{align*}
&\quad \E_{R \in \SO(n)} \E_{g \in \Pin(n)} \norm*{\Phi_A A_R(g) - \left( \Gamma(g) \ot I \right) \Phi_A}_{\rho_A}^2 \\
&= \E_{h \in \Pin(n)} \E_{R \in \SO(n)} \left\lVert \Phi_A A(R x_1(h)) \cdots A(R x_{k(h)}(h)) - \left( \gamma(R x_1(h)) \cdots \gamma(R x_{k(h)}(h)) \ot I \right) \Phi_A \right\rVert_{\rho_A}^2 \\
&\leq \bigO(n^3 \eps).
\end{align*}
This is the desired inequality.
\end{proof}

\section{Effective Gowers--Hatami Theorem for Compact Groups}\label{sect:bonus}
In this section, we present a version of the Gowers--Hatami theorem~\cite{GowersHatami2017} that combines a number of appealing features: it is direct and explicit, applies to finite and compact ``continuous'' groups alike, and clarifies what is required so that the isometry can be implemented efficiently.
We also give a version that allows the extraction of a specific irreducible, which is important for self-testing.
While we claim no novelty, to the best of our knowledge this proof has not been presented in this way before in the literature.
The special case of finite groups was described in \cite{samplingpaper}.
As noted therein, the argument is distinct from the one given recently in \cite{metger2024succinct} (which obtains the isometry implicitly) and also from the general argument \cite{vidicknotes} (we postpone the use of the non-commutative Fourier transform, which we believe clarifies the isometries' interpretation).
We present the argument here since it fits the theme of the paper, we expect that it may prove useful for self-testing of other continuous groups, and to keep the paper fully self-contained.

Let $G$ be a compact (Hausdorff) topological group, equipped with its unique normalized bi-invariant Haar measure.
We write $\E_{g \in G}$ for the corresponding Haar average.
In the following, all functions on $G$ will be assumed to be suitably measurable (in particular, operator-valued functions and their adjoints are assumed to be Bochner measurable after application to any fixed vector).
Let $\mathcal H$ be a Hilbert space.
We denote by $\Ell^2(G, \calH)$ the Hilbert space of square-integrable $\mathcal H$-valued functions on~$G$.
We abbreviate $\Ell^2(G) \coloneqq \Ell^2(G, \C)$.
There is a unitary representation of~$G$ on $\Ell^2(G, \calH)$ given by right multiplication in the argument:
\begin{align*}
    R \colon G \to \U(\Ell^2(G, \calH)), \quad (R(h) \varphi)(g) \coloneqq \varphi(g h).
\end{align*}
For a quantum state~$\rho$ on~$\mathcal H$ and any other Hilbert space~$\mathcal K$, we will also use the state-weighted (positive semidefinite) inner product~$\ip{A}{B}_\rho := \tr (A^\dagger B \rho)$ on the space of bounded linear operators $\mathcal H \to \mathcal K$, with associated (semi)norm~$\norm{A}_\rho$, as in \Cref{eq:state-dependent-seminorm}.
We observe that $\norm{W A}_\rho \leq \norm{A}_\rho$ for any contraction~$W \colon \mathcal K \to \mathcal L$ of Hilbert spaces (i.e., $\norm{W} \leq 1$ or $W^\dagger W \leq I$).

The following variant of the Gowers--Hatami theorem shows that an approximate representation can always be rounded into an exact one on a larger space, by a simple universal formula, with pointwise (worst-case) guarantees that hold for any state-weighted norm.

\begin{theorem}\label{thm:GH}
For any function $f \colon G \to \mathrm U(\mathcal H)$, the linear map
\begin{align*}
    V \colon \calH \to \Ell^2(G, \calH), \quad (V \ket\psi)(g) \coloneqq f(g) \ket\psi
\end{align*}
is an isometry.
Moreover, for every $h \in G$ and any state $\rho$ on $\mathcal H$:
\begin{align*}
  \norm{ V f(h) - R(h) V }_\rho^2 = \E_{g \in G} \norm{ f(g) f(h) - f(gh) }_\rho^2.
\end{align*}
\end{theorem}
\begin{proof}
Because each $f(g)$ is unitary, it is immediate that $V$ is an isometry.
It remains to prove the norm equality.
To this end, we note that
\[ \left( \left( V f(h) - R(h) V \right)\ket\psi \right)(g) = \left( f(g) f(h) - f(gh) \right) \ket\psi \]
for all $g, h \in G$ and $\ket\psi\in\calH$, hence
\[ \left( V f(h) - R(h) V \right)^\dagger \left( V f(h) - R(h) V \right) = \E_{g \in G} \left( f(g) f(h) - f(gh) \right)^\dagger \left( f(g) f(h) - f(gh) \right), \]
and now the claim follows by taking the trace with $\rho$.
\end{proof}

Despite the simplicity of the proof, the theorem has a direct algorithmic interpretation.
Note that we can canonically identify $\Ell^2(G,\mathcal H) \cong \Ell^2(G) \ot \mathcal H$.
As is plain from the formula, the isometry~$V$ in \cref{thm:GH} can be implemented by the composition of the following two subroutines:
\begin{enumerate}[noitemsep]
\item
Add a register with Hilbert space~$\Ell^2(G)$ initialized with the constant function everywhere equal to one (denote $\ket u$ below).
When $G$ is a finite group, this is simply the uniform superposition of all group elements.
\item Apply the \emph{controlled} approximate representation, which is the unitary $U_f$ on $\Ell^2(G, \calH)$ given by $(U_f \varphi)(g) = f(g) \varphi(g)$.
\end{enumerate}
Schematically as a quantum circuit:
\begin{align*}
    V = \begin{quantikz}
        & \wireoverride{n} \gategroup[2,steps=3,style={rounded corners, inner xsep=2pt},background,label style={label position=below,anchor=north,yshift=-0.2cm}]{} &\wireoverride{n} \lstick{$\ket u$} & \ctrl{1} & \\
        & & & \gate{f} &
    \end{quantikz}
\end{align*}
In particular, $V$ can be implemented efficiently in any setting where the uniform superposition over all group elements and the controlled approximate representation can be implemented efficiently.

In applications one is often interested in rounding towards a \emph{specific} irreducible representation---say, the spin representation of the Pin group, its restriction to the Clifford subgroup, or the defining representation of the Weyl-Heisenberg group.
This can be achieved if the representation satisfies a natural hypothesis that generalizes, e.g., the condition that an approximate representation of the Weyl--Heisenberg group generated by the Pauli X and Z operators satisfies $f(-P) = -f(P)$ for all elements~$P$.

\begin{theorem}\label{thm:GH irrep}
Let $\lambda \colon G \to \U(\mathcal H_\lambda)$ be a continuous irreducible representation with character~$\chi_\lambda$ and dimension~$d_\lambda$.
Let $f \colon G \to \mathrm{U}(\mathcal H)$ be a function such that $d_\lambda \E_{h \in G} \overline{\chi_\lambda(h)} f(gh) = f(g)$ for all~$g \in G$.
Then the image of the isometry $V$ from \cref{thm:GH} is in the $\lambda$-isotypical subspace of~$\Ell^2(G,\calH)$, which using the Fourier transform of~$G$ identifies with $\End(\calH_\lambda) \ot \calH \cong \calH_\lambda \ot \calH_\lambda^* \ot \calH$.
Thus, the corresponding restriction
\begin{align*}
    V_\lambda \colon \calH \to \End(\calH_\lambda) \ot \calH \cong \calH_\lambda \ot \calH_\lambda^* \ot \calH, \quad V_\lambda \ket\psi = \sqrt{d_\lambda} \E_{g \in G} \lambda(g)^\dagger \ot f(g) \ket\psi
\end{align*}
is an isometry and it satisfies, for every $h \in G$ and any state $\rho$ on $\mathcal H$:
\begin{align*}
     \norm{ V_\lambda f(h) - \bigl( \lambda(h) \ot I_{\calH_\lambda^*} \ot I_{\calH} \bigr) V_\lambda }_\rho^2 = \E_{g \in G} \norm{ f(g) f(h) - f(gh) }_\rho^2.
\end{align*}
\end{theorem}
\begin{proof}
The orthogonal projection on the $\lambda$-isotypical subspace of a representation~$R$ is given by $P_\lambda = d_\lambda \E_{g \in G} \overline{\chi_\lambda(g)} R(g)$.
Let $V$ denote the isometry from \cref{thm:GH}.
Then:
\begin{align*}
    (P_\lambda V \ket\psi)(g)
=  d_\lambda \E_{h \in G} \overline{\chi_\lambda(h)} f(gh) \ket\psi
= f(g) \ket\psi
= (V \ket\psi)(g) \qquad (g \in G, \ket\psi\in\calH).
\end{align*}
Thus, $P_\lambda V = V$, which confirms that the image of $V$ is in the $\lambda$-isotypical subspace.
To obtain a concrete formula for the restriction, recall the Peter--Weyl theorem:
there exists a unitary, called the \emph{Fourier transform} of the group~$G$, that identifies
\[ \mathcal F \colon \Ell^2(G) \to \bigoplus_{\lambda \in \Lambda} \End(\calH_\lambda) \cong \bigoplus_{\lambda \in \Lambda} \calH_\lambda \ot \calH_\lambda^*, \]
where $\Lambda$ is a complete set of pairwise-inequivalent irreducible unitary representations of~$G$, and the representation~$R(h)$ on $\Ell^2(G)$ corresponds to \emph{left} multiplication by $\lambda(h)$ on each $\End(\calH_\lambda)$ or, equivalently, by $\lambda(h) \ot I_{\calH_\lambda^*}$ on each $\calH_\lambda \ot \calH_\lambda^*$.
Thus, the $\lambda$-isotypical component corresponds to the $\lambda$-direct summand.
Concretely, a function $\varphi$ is mapped to the collection $\hat\varphi = (\hat\varphi(\lambda))_{\lambda\in\Lambda}$, where $\hat\varphi(\lambda) \coloneqq \sqrt{d_\lambda} \E_{g \in G} \varphi(g) \lambda(g)^\dagger$; cf.~\cite{vidicknotes}.
The above generalizes directly to $\Ell^2(G, \calH) \cong \Ell^2(G) \ot \calH$, resulting in the displayed formula for $V_\lambda$ by composing $V$, the Fourier transform, and the restriction to the $\lambda$-direct summand.
Now the norm equality follows from \cref{thm:GH}.
\end{proof}

As described in \cref{thm:GH irrep}, the isometry $V_\lambda$ can be implemented by first implementing~$V$ as described above, then applying the Fourier transform for the group~$G$ and finally discarding the irrep label (which will always be $\lambda$).
Schematically as a quantum circuit:
\begin{align*}
    V_\lambda = \begin{quantikz}[row sep=0.1cm]
        \setwiretype{n} &  \gategroup[4,steps=6,style={rounded corners, inner xsep=2pt},background,label style={label position=below,anchor=north,yshift=-0.2cm}]{} &  &  & \gate[3]{\mathcal F} & \setwiretype{q} \rstick{$\ket\lambda$} \\
        & \wireoverride{n} & \wireoverride{n} \lstick{$\ket u$} & \ctrl{2} & & & & \\
        & \wireoverride{n} & \wireoverride{n} & \wireoverride{n} & \wireoverride{n} & & & \\
        & & & \gate{f} & & & &
    \end{quantikz}
\end{align*}
In particular, $V_\lambda$ can be implemented efficiently in any setting where the uniform superposition over the group can be efficiently prepared and the Fourier transform of the group (if the inverse is also available, this can be used to realize the uniform superposition) and the controlled approximate representation can be implemented efficiently.

\section*{Acknowledgments}
AK, GM, SS and MW acknowledge support by the Deutsche Forschungsgemeinschaft (DFG, German Research Foundation) under Germany's Excellence Strategy - EXC 2092 CASA - 390781972.
AK and MW acknowledge support by the Deutsche Forschungsgemeinschaft (DFG, German Research Foundation) under Germany's Excellence Strategy~--~EXC-2111~--~390814868.
GM acknowledges support by the European Union (ERC Grant ObfusQation, 101077455).
MW acknowledges support by the European Union (ERC Grant SYMOPTIC, 101040907), by the Deutsche Forschungsgemeinschaft (DFG, German Research Foundation, 556164098), by the German Federal Ministry of Research, Technology and Space (QuSol, 13N17173), and by the Klaus Tschira Foundation.

\bibliographystyle{alpha}
\bibliography{references}
\end{document}